	\providecommand\FINAL[1]{} 

\newcommand\SOURCE[1]{\FINAL{\marginpar{source: #1}}}

\newcommand\SHORTER[1]{}

\newcommand\parametredeapxproof{}
	\renewcommand\parametredeapxproof{appendix=inline}

\providecommand\FINAL[1]{#1}%
\providecommand\PASFINAL[1]{#1}
\FINAL{\renewcommand\PASFINAL[1]{} }

\newcommand{\cp}[1]{\operatorname{#1}}
\newcommand{\Ptime}{\cp{PTIME}} 
\newcommand\PTIME{\Ptime}

\newcommand{\Pspace}{\cp{PSPACE}}
\newcommand\PSPACE{\Pspace}

\newcommand{\FPspace}{\cp{FPSPACE}}

\newcommand{\PLS}{\cp{FLS}}
\newcommand{\PPAD}{\cp{PPAD}}

\newcommand\RCA{\cp{RCA}}
\newcommand\WKL{\cp{WKL}}
\newcommand\ACA{\cp{ACA}}

\documentclass{article}
\AtBeginDocument{%
  }

    \usepackage{todonotes}

\usepackage[\parametredeapxproof]{apxproof}

\usepackage{amsmath}
\usepackage{amsfonts}
\usepackage{amssymb}
\usepackage[capitalise]{cleveref}
\usepackage{xfrac}
\usepackage{authblk}
\usepackage{geometry}

\makeatletter
\newcommand\footnoteref[1]{\protected@xdef\@thefnmark{\ref{#1}}\@footnotemark}
\makeatother

\newcommand{\norm}[1]{\left\lVert#1\right\rVert}

\usepackage{soul,xcolor}

\newtheoremrep{theorem}{Theorem}
\newtheoremrep{lemma}{Lemma}
\newtheoremrep{proposition}{Proposition}
\newtheoremrep{corollary}{Corollary}
\theoremstyle{definition}
\newtheoremrep{definition}{Definition}
\newtheoremrep{example}{Example}
\theoremstyle{remark}
\newtheoremrep{remark}{Remark}

\begin{document}

\setstcolor{red}
\title{Relating the Computational and Logical Difficulty of Solving ODEs: From Polynomial to Discontinuous Right-Hand Sides}

\author[ ]{Olivier Bournez}

\author[ ]{Alonso N\'u\~nez}

\affil[ ]{bournez@lix.polytechnique.fr}
\affil[ ]{herreranunez@lix.polytechnique.fr}
\affil[ ]{ }
\affil[ ]{LIX, École polytechnique\\ Palaiseau,France}

\maketitle

\begin{abstract}
When a computer algebra system fails to solve an Ordinary Differential
Equation, is this a limitation of its implementation, or a genuine
computational barrier?

Three traditions bear on the question of how hard it is to solve
ODEs. Modern algorithms in computer algebra can be extremely efficient: for instance, 
Newton-type methods solve polynomial ODEs
over $\mathbb{Q}[[X]]$ in quasi-linear time. Work on analog models of
computation has shown that polynomial ODEs and Turing machines are
two presentations of the same phenomenon, with solution length acting
as time and precision as space. Computable analysis shows that ODEs
can be intrinsically hard, namely undecidable, even
$\mathsf{PSPACE}$-complete, over compact domains. Comparing these
traditions is natural and necessary, yet such comparisons routinely
reduce to comparisons of encodings rather than of underlying
algorithmic content.

We argue that reverse mathematics is a relevant tool in this
setting: it provides a representation-invariant lens in which
algorithmic content is compared directly. 
We prove that every level
of the Big Five hierarchy of reverse mathematics is inhabited by a natural statement from
classical ODE theory, as an exact equivalence rather than an analogy:
the regularity of the vector field $f$ is an intrinsic algorithmic
invariant placing the initial value problem $y'(t)=f(t,y(t))$,
$y(t_0)=y_0$, into one of several computational strata, ranging from
polynomial-time solvability to transfinite computation.

The resulting stratification acts as a practical diagnostic common to
the three traditions. By abstracting from representation, it separates
fundamental barriers from the technical shortcomings of symbolic
solvers, the representation artefacts of analog encodings, and the
effectivity constraints of computable analysis, identifying the
intrinsic parameters (length bounds, radii of convergence, moduli of
continuity) under which feasibility is restored.
\end{abstract}

\section{Introduction}

\SOURCE{Cl sur abstract}

Computer algebra systems such as Maple, Mathematica, and SageMath routinely solve ordinary differential equations with well-behaved right-hand sides. For an initial value problem
\begin{equation}\label{eq:ivp}
y'(t)=f(t,y(t)),\quad y(t_0)=y_0,\tag{IVP}
\end{equation}
they often produce explicit formulas, series expansions, or other symbolic representations when $f$ is smooth or Lipschitz; see \cite{bostan2017algorithmes} for a comprehensive treatment. In favorable cases the toolkit is strikingly efficient: Newton-type methods compute power-series coefficients in quasi-linear time \cite{brent1978fast,bostan2017algorithmes}, dramatically outperforming direct Picard iteration.

Independently, comparisons between analog and digital computation have uncovered a deep correspondence: polynomial ODEs and Turing machines are two presentations of the same computational phenomenon (see Section~\ref{sec:explainodeandMT}). Arbitrary Turing computations embed into finite-dimensional polynomial flows, and conversely such dynamics encode any computable process~\cite{JOC2007}. This extends to complexity: solution length acts as time~\cite{JournalACM2017,TheseAmaury}, while the precision needed to distinguish outcomes mirrors space~\cite{Icalp2024,TheseManon}. Symbolic ODE solving thus becomes a form of program analysis, and lower bounds from computability and complexity theory translate into intrinsic limitations on any solver (e.g.\ hardness of computing radii of convergence~\cite{dsg06a}).

Computable analysis sharpens the picture~\cite{Wei00,Ko91,brattka2008tutorial}:
over compact domains, ODEs with computable right-hand sides may admit
no computable solution~\cite{Ko83,PR79,Abe71}, and restricting to
polynomial-time computable dynamics does not restore tractability---
solving can remain $\mathsf{PSPACE}$-complete~\cite{kawamura2009lipschitz,Ko83}.

A recurring methodological issue runs across all these results: in each
community (computer algebra, computable analysis, complexity theory,
reverse mathematics, Weihrauch reducibility) representations and
complexity models are part of the claim, not neutral wrappers, and
cross-community comparisons easily devolve into comparisons of
encodings rather than of the underlying algorithmic tasks.
Richardson’s theorem~\cite{Richardson1968} is emblematic. This paradigmatic undecidability result in computer algebra is fundamentally representation-driven: although Hilbert’s tenth problem can be readily embedded into questions about dynamical systems~\cite{Ruo97b,CIEChapter2007}, such \emph{static undecidability} does not, in itself, entail any intrinsic difficulty in predicting the resulting dynamics~\cite{Ruo97b}. 

We argue in this work that \emph{reverse mathematics} provides an appropriate representation invariant framework to reconcile efficient algorithms, hardness results, and representation issues. Reverse mathematics classifies theorems by the axioms needed to prove them, and its standard subsystems come with a robust algorithmic interpretation: $\RCA_0$ captures basic effective constructions, $\WKL_0$ compactness-based search, $\ACA_0$ arithmetical comprehension, $\mathrm{ATR}_0$ transfinite recursion, and $\Pi^1_1\text{-}\mathrm{CA}_0$ strong analytical comprehension.

We establish that every level of the Big Five hierarchy is inhabited by a natural ODE statement from classical analysis, and the correspondence is exact: the differential statement and the logical principle are mutually derivable, hence carry the same computational content. The classification turns familiar theorems into explicit resource bounds while still exposing the quantitative hypotheses symbolic methods exploit in practice---compactness of the domain, moduli of uniqueness/continuity, radii of convergence. Each level is a phase transition: membership yields constructive upper bounds, separation yields lower bounds.

Concretely, we obtain a strict computational\footnote{We use \emph{algorithmic complexity} as an umbrella term for computability and complexity theory.} \emph{hierarchy} for~\eqref{eq:ivp}, where the regularity of $f$ serves as an \emph{algorithmic invariant} determining the stratum, ranging from an efficient polynomial-time stratum (polynomial dynamics with quantitative data) through Lipschitz, Osgood, and continuous regimes up to transfinite computation for discontinuous but solvable right-hand sides.

\newcommand\Encode{\mathrm{Encode}}
\newcommand\Decode{\mathrm{Decode}}
\newcommand\NextConfiguration{\mathrm{NextConf}}

\section{Relating computations and (polynomial) ODEs}
\label{sec:explainodeandMT}

We briefly recall how programs and polynomial ODEs can simulate each
other. This two-way translation, initiated in \cite{dsg05,JOC2007} in the
context of continuous-time analog models, provides a concrete basis for
viewing ODE solving as a computational process.

\paragraph{From programs to polynomial ODEs.}
Any program computing a (possibly partial) function $f$ iterates a
transition map until a halting condition holds, yielding the so-called Kleene normal form:
\begin{equation}\label{eq:kleene-nf}
  f(x)=\Decode\bigl(\mu t.\, M(\Encode(x),t)=0\bigr),
\end{equation}
for primitive-recursive $\Encode,\Decode,M$, where $\mu$ denotes
unbounded minimization (for a Turing machine: $\Encode(x)$ is the initial
configuration, $M(c,t)$ the configuration after $t$ steps from $c$,
$M(\Encode(x),t)=0$ expresses halting at time $t$, and $\Decode$ reads
the output). 

The discrete dynamics can then be expressed as a primitive-recursive
recurrence
\begin{equation}\label{eq:discrete-iter}
\begin{aligned}
  M(x,0)   &= \Encode(x),\\
  M(x,t+1) &= \NextConfiguration\!\bigl(M(x,t)\bigr).
\end{aligned}
\end{equation}
with a standard encoding identifying configurations with~$\mathbb{R}^d$
(e.g.\ $d=3$: one coordinate for the state, two reals for the left/right
tape contents via base-$k$ expansions). The iteration~\eqref{eq:discrete-iter}
can then be embedded into a continuous-time flow of a polynomial vector field with
rational coefficients,
\begin{equation}\label{eq:poly-ode}
  y'(t)=P\bigl(y(t)\bigr),
\end{equation}
whose solution tracks the discrete computation at integer times via a
clock variable and periodic update gadgets. The simulation can be made \emph{robust} to small perturbations in the initial condition and the dynamics~\cite{dsg05,JOC2007}.

\paragraph{From polynomial ODEs to programs.}
Conversely, given $P$ and an initial condition, classical numerical
schemes (Picard, Taylor, validated step-based procedures under a priori
bounds) produce approximations on any interval of existence and
uniqueness. 
\paragraph{More than universality.}
This correspondence between programs and polynomial ODEs should not be
read as yet another universality statement. Here the simulation is
\emph{direct}: the algorithmic state is represented by the ODE state,
and the flow mirrors step-by-step execution. Beyond Turing
completeness, this enables intrinsic geometric quantities to serve as
computational resources, accessed directly from the solution itself rather than recovered from an encoding. Polynomial-time computations correspond to polynomial ODEs whose relevant solution curves have polynomial length,
and vice versa
\cite{bournez2012complexity,JournalACM2017,TheseAmaury}; under
robustness, the precision separating outcomes mirrors space
\cite{Icalp2024,TheseManon}. This viewpoint underlies several recent
works, e.g. \cite{StacsBournezGozzi2024,Computability2024Gozzi},
which relate the complexity of solving discontinuous ODEs to
transfinite computation, as well as classical characterizations of
complexity classes via continuous dynamical systems
\cite{TheseAmaury,TheseRiccardo} and via discrete difference schemes
\cite{MFCSJournal,TheseManon,Antonelli0K24,AntonelliDK26,Antonelli0K25}.

\paragraph{Why do we do this?}
Our goal is not to promote a single model of computation, but to
isolate the structural principles that govern computational difficulty
across discrete, continuous, and higher-order settings, and accross various traditions.

Solving an ODE is, by its very nature, a \emph{search} problem: one
must find a function satisfying a given differential constraint.
Running a Turing machine is likewise a search, for a valid execution
trace consistent with the transition rules. The correspondence above
shows that these are two presentations of the same phenomenon and
that efficiency is preserved across the translation: finding an
efficient algorithm amounts to finding a, possibly alternative, way to
carry out the search. The central question is then \emph{what makes a
search efficient?}.

We state that reverse mathematics answers precisely this: it abstracts away from
representation and classifies search problems by the logical
principles they require, such as compactness, determinism,
arithmetical comprehension, and iterated limits, each with a clear
computational interpretation. 

As evidence, we show that every level of
the Big Five is naturally realised by structural theorems about ODEs.
Since these systems are pairwise distinct and reflect genuine
computational resources, this yields a sharp hierarchy of algorithmic
difficulty indexed by regularity assumptions. Crucially, the resulting
comparisons transcend representation issues, unlike all the other
complexity comparisons we have encountered.

\section{Reverse Mathematics}\label{sec:rm}

Mathematicians routinely say that one theorem is \emph{stronger} than
another, or that two are \emph{equivalent}. Taken literally, this is
misleading: any two true statements are equivalent in full logic.
Reverse mathematics fixes a \emph{weak} base theory $W$ and asks which
additional principles a theorem requires. We write $W\vdash\sigma$ when
the statement $\sigma$ is derivable from the axioms of $W$, and
$W\nvdash\sigma$ otherwise. Over $W$, theorems $\sigma,\tau$ are
\emph{equivalent} if $W\vdash\sigma\leftrightarrow\tau$, while $\sigma$
is \emph{stronger} than $\tau$ if $W\vdash\sigma\to\tau$ but
$W\nvdash\tau\to\sigma$. The base $W$ must be weak enough for the
comparison to be nontrivial.

The standard arena of classical reverse mathematics is
\emph{second-order arithmetic}, a two-sorted language $\mathcal{L}_2$
with numbers and sets of numbers, expressive enough to encode real
numbers, functions, and much of analysis. A model
$$M=(|M|,\mathcal{S}_M,0,1,{+},{\cdot},{<})$$ has a set sort
$\mathcal{S}_M\subseteq\mathcal{P}(|M|)$; we write $M\models\sigma$ when
$\sigma$ holds in $M$. When $|M|=\omega$ with the standard operations,
$M$ is called an \emph{$\omega$-model}, determined by $\mathcal{S}_M$;
the canonical example $\mathrm{REC}$ has $\mathcal{S}_{\mathrm{REC}}$
equal to the computable sets.

We write $A\leq_T B$ for \emph{Turing reducibility}: $A$ is computable
by a Turing machine allowed to consult~$B$ as an oracle, i.e.\ to query
membership of any integer in~$B$ during its computation. A set $C$ is
computable \emph{relative to}~$B$, or \emph{relativized to}~$B$, if
$C\leq_T B$. We write $0'$ for the \emph{halting set}---the set of
codes of Turing machines that halt on the empty input, classically
noncomputable---and $A'$ for the \emph{Turing jump} of~$A$, namely the
halting set relative to an $A$-oracle machine. Thus $0'$ measures the
intrinsic complexity of halting, and each jump $A\mapsto A'$ strictly
increases complexity.

These computability notions fit inside reverse mathematics via the
$\omega$-model correspondence: a subsystem of second-order arithmetic
is determined, over~$\omega$, by the sets~$\mathcal{S}_M$ it
admits---computable, computable-from-an-oracle, arithmetically
definable, and so on. Stronger comprehension or recursion principles
admit more sets, hence more theorems become provable. 

Reverse mathematics was initiated by Friedman
\cite{friedman1975some,friedman1976systems} and subsequently
systematised by Simpson \cite{simpson2009subsystems}, who established
that five subsystems suffice to classify most theorems of analysis and
combinatorics: $\RCA_0$, $\WKL_0$, $\ACA_0$, $\mathrm{ATR}_0$, and
$\Pi^1_1\text{-}\mathrm{CA}_0$ (the \emph{Big Five}). These form a
strictly increasing chain
\[
\RCA_0 \subsetneq \WKL_0 \subsetneq \ACA_0 \subsetneq \mathrm{ATR}_0
\subsetneq \Pi^1_1\text{-}\mathrm{CA}_0,
\]
where each inclusion is proper: every theorem provable at a lower
level is provable at any higher level, but strict separations are
witnessed by $\omega$-models in which the weaker theory holds and the
stronger one fails. Among them, the \emph{Weak König's Lemma} (WKL) states
that every infinite binary tree has an infinite path; $\WKL_0$ is
$\RCA_0$ augmented with WKL. For ODEs, \cite{Simpson84Peano} showed
that Cauchy-Peano $\equiv\WKL_0$ and that Cauchy-Lipschitz is
provable in $\RCA_0$. We prove that all five levels are inhabited by
natural ODE statements.

\subsection{An example of application to ODEs}\label{sec:rm-example}

Reverse mathematics gives a high-level lens on \emph{why} solving ODEs
can be hard, separating genuine difficulty from artefacts of
representation.

\paragraph{Impossibility results for free.}
Many existence theorems in analysis have the $\Pi^1_2$ form
$$\sigma\equiv\forall U\,\exists V\,\theta(U,V),$$ where $U$ is input
data (e.g.\ a continuous function), $V$ is the asserted object (e.g.\ a
solution), and $\theta$ is arithmetical. Peano's existence theorem fits
this schema: for every continuous~$f$ there exists a solution~$y$ to
the IVP~\eqref{eq:ivp}.

Suppose $\RCA_0\vdash\sigma\leftrightarrow\WKL$, so that any proof of
$\sigma$ requires a principle equivalent to WKL. Consider the
$\omega$-model $\mathcal{M}_{\mathrm{rec}}$ of recursive sets. It
satisfies $\RCA_0$ but not $\WKL$: there exist infinite recursive
binary trees with no recursive infinite path. Hence
$\mathcal{M}_{\mathrm{rec}}\models\neg\sigma$. Unwinding what this
means yields a \emph{computable counterexample}: some recursive input
$U$ admits no recursive $V$ with $\theta(U,V)$. Classically, solutions
exist for every input; but for this particular computable input, no
computable solution exists. The phenomenon is no accident: it is a
direct logical consequence of $\sigma$ having exactly the strength
of~$\WKL$.

\begin{example}
This is the conceptual core of the classical computable-analysis
results showing that some computable ODEs admit no computable solution
\cite{PR79,Abe71}, here abstracted away from the technical details of
the constructions.
\end{example}

\paragraph{Impossibility of hardness.}
Reverse mathematics also rules out \emph{overly strong} hardness. Call
$A$ \emph{low} if $A'\leq_T 0'$: low sets, though possibly
noncomputable, are computationally weak in that they cannot compute
the halting problem (see \cite{Soare2016} for a textbook account).

\begin{lemmarep}\label{lem:low}
If $A$ is low, then $0'\not\leq_T A$.
\end{lemmarep}
\begin{proof}
If $0'\leq_T A$, then $0''=(0')'\leq_T A'\leq_T 0'$, contradicting the
strictness of the jump.
\end{proof}

By the Low Basis Theorem \cite{jockusch1972degrees}, every infinite
recursive tree has a path of low Turing degree. Hence whenever
$\sigma\equiv\WKL_0$, every computable $U$ admits a low-degree solution
$V$, from which (by Lemma~\ref{lem:low}) the halting problem is
unreachable. So solutions of computable ODEs may be noncomputable, but
never enough to encode halting: the noncomputability lives entirely in
the weak choice principle of $\WKL_0$.

Both conclusions followed from the single fact that Peano's theorem is
$\WKL_0$, abstracting away from representation.

\subsection{Going to complexity?}\label{sec:brm}

Classical reverse mathematics is insensitive to resources: all
computable objects look the same. \emph{Bounded reverse mathematics} and
\emph{bounded arithmetic} restore sensitivity by restricting
definability and induction so that provably total functions coincide
with complexity classes \cite{krajicek1995bounded,cook2010logical,ReverseMathOfComplexityLowerBournds24}.
We do not use bounded RM machinery here, but borrow its lesson from
implicit complexity: \emph{complexity lives in the recurrence schema}.
The question is not whether an iteration exists, but whether it can be
reorganized so that both iteration count and intermediate values remain
polynomially bounded.

\begin{example}[Iteration: computability vs.\ complexity]
The Kleene normal form~\eqref{eq:kleene-nf} is insensitive to how large
$t$ is. Efficiency requires a polynomial bound:
\begin{equation}\label{eq:discrete-iter:bis}
f(x)=\Decode\bigl(M(\Encode(x),P(|x|))\bigr),
\end{equation}
for some polynomial $P$ and $|x|$ the binary length. This is a change
of variable: the iteration count is replaced by a quantity polynomial
in input size, exactly the kind of reshaping made explicit in implicit complexity.
\end{example}

\paragraph{Recursion on notation.}
In Cook-style polynomial-time arithmetic $\mathrm{PV}$, polynomial-time
functions are captured by \emph{recursion on notation}:
\begin{equation}\label{sec:impl}
f(0)=a,\quad f(2x)=g(x,f(x)),\quad f(2x+1)=h(x,f(x)),
\end{equation}
with $g,h$ polynomial-time. Recursion is controlled by the binary
length of the input, enforcing feasible iteration machine-independently.

\subsubsection{Applications to complexity of solving ODEs}
\label{sec:brm-eff-odes}

\paragraph{A guiding obstruction.}
The function $E(x)=2^x$ cannot be expressed using bounded recursion on
notation, the basic schema underlying polynomial-time computation \cite{Cob65,cook2010logical}.

\paragraph{Polynomial ODEs and iteration depth.}
Consider the polynomial IVP~\eqref{eq:poly-ode} with
$P\in\mathbb{Q}[X]$ (scalar case for simplicity). A classical method is
Picard iteration: starting from $u_0(t)\equiv y_0$,
\(
u_{n+1}(t)=y_0+\int_0^t P(u_n(s))\,ds.
\)
Each $u_n$ is a polynomial (or truncated power series), and the
definition is primitive recursion over $n$.

Efficient algorithms instead use Newton-Kantorovich iteration on
$F(y)=y'-P(y)=0$ (Newton ``lifting''). Given $y_k$, compute $\delta_k$
from the linearized equation
\(
\delta_k'(t)-P'(y_k(t))\,\delta_k(t)=-(y_k'(t)-P(y_k(t))),\ \delta_k(0)=0,
\)
and update $y_{k+1}=y_k+\delta_k$. Computing modulo $t^{N_k}$, the
update works modulo $t^{2N_k}$: the number of correct coefficients
doubles at each step. Picard is sequential in precision ($O(n)$ steps
for error $2^{-n}$); Newton is not ($O(\log n)$ steps). Both are
primitive-recursive in form; only the rate of convergence differs.
Newton yields a polynomial-time method for the solution \emph{as a
formal power series}.

\paragraph{Analytic continuation and sequentiality.}
To approximate the real solution on $[0,T]$, one evaluates the series.
Let $R$ be the radius of convergence at $0$. When $T<R$, Cauchy
estimates give truncation order polynomial in $n$ and in
$\log(1/(1-T/R))$; as $T\to R$ the cost blows up. For $T$ beyond $R$,
analytic continuation proceeds via patches
$0=t_0<t_1<\cdots<t_m=T$, inherently \emph{sequential in the number of
patches $m$}. Newton reduces the cost within each patch but does not
remove the dependence on $m$.

\paragraph{From intrinsic parameters to input-size complexity.}
Can we devise a polynomial algorithm for~\eqref{eq:poly-ode} in the
sense of computable analysis, i.e.\ outputting a $2^{-n}$-approximation
of $y(T)$ in time polynomial in $n$ and in the binary sizes of $T$ and
$y_0$ \cite{Ko91}? The previous discussion yields primitive-recursive
solutions in $n$, but infeasibility in the size of the other arguments
when they take integer values. Assume for contradiction that such a
uniform solver exists. Taking $P(u)=u$, it yields a procedure for
$E(x)=2^x$ over the integers, expressing $E$ as a bounded recursion on
notation---impossible. The point is structural, not specific to~$E$: a
polynomial-time solver in input size would express some function in the
form~\eqref{sec:impl}, contradicting implicit-complexity bounds.

This illustrates the general point: implicit complexity distinguishes
algorithms polynomial in intrinsic analytic parameters (length, radii of
continuation) from those polynomial in the discrete input size, and
explains why moving from one to the other is provably impossible without
additional structure.

\begin{remark}
This resonates with the parameterized approach in computable analysis
\cite{kawamura2015computational,kawamura2018parameterized}, where
complexity is measured relative to analytic parameters (moduli,
bounds, radii). Many efficiency barriers can be explained at the level
of recursion schemes, independently of representation.
\end{remark}

\section{Framework and statements of results}
\label{sec:statements}

\subsection{State of the art.}
Our work connects several strands of research that have largely
developed independently: classical ODE theory, computable analysis,
analog models of computation, and logical classifications of
computational strength.

Classical existence and uniqueness results for ODEs, from Cauchy and
Peano to Lipschitz and Osgood \cite{Hartman,CL72}, are traditionally
formulated as analytic statements. Our contribution is to make
explicit their \emph{algorithmic content}: which computational
principles they justify, and which they provably do not.

Solving an ODE can be seen as computing a fixed point of an operator
such as the Picard operator. Fixed-point constructions are central
in logic, semantics, and analysis: in descriptive complexity, least
and greatest fixed points characterise $\PTIME$ and $\PSPACE$
\cite{Imm99,Var82}, and continuous fixed-point problems span a wide
spectrum, from polynomial time to $\PLS$- and $\PPAD$-completeness
\cite{Pap94b,Daskalakis06}. In computable analysis, even a unique
fixed point may be noncomputable \cite{Ko91,PORI17}, while
approximate ones are always computable with complexity governed by
moduli of continuity. Our classification identifies which fixed-point
principles---deterministic, compactness-based, or limit-based---are
required at each regularity level.

Polynomial differential equations have long been studied as
continuous-time models of computation; see the survey
\cite{bournez2021survey}. The computability and complexity of ODE
solving have been extensively studied in computable analysis
\cite{Ko91,GracaZhongHandbook}, which extends classical computability
to $\mathbb{R}^n$ via the Type-2 Theory of Effectivity
\cite{Wei00,brattka2008tutorial}. On compact domains, non-uniqueness
leads to noncomputability \cite{Ko83,PR79,Abe71}, while uniqueness
restores computability \cite{Ruo96,collins2008effectivesimpl,collins2009effective}
but allows arbitrarily high complexity \cite{Ko91,Mi70}. Under
Lipschitz assumptions, solutions are $\FPspace$-computable and
$\Pspace$-completeness can already arise
\cite{kawamura2009lipschitz}, even for $\mathcal{C}^\infty$ vector
fields \cite{kawamura2014computational}; polynomial-time solvability
on compact domains would thus collapse $\Ptime$ and $\Pspace$. Over
non-compact domains, polynomial ODEs simulate Turing machines
\cite{dsg06a}, yielding undecidability for very simple analytic
systems. Controlling growth recovers tractability: this motivates
parameterised approaches bounding solution growth
\cite{bournez2012complexity}, later refined by a single intrinsic
parameter, the length of the solution curve \cite{PoulyGraca16},
underpinning the principle \emph{time complexity = length}.

Our work adopts the perspective of \emph{reverse mathematics}
\cite{simpson2009subsystems}, where Simpson showed that Cauchy-Peano is
equivalent to $\WKL_0$ and that Cauchy-Lipschitz is provable in $\RCA_0$
\cite{Simpson84Peano}. The arguments we expand in
Section~\ref{sec:rm-example} are drawn from a short discussion
in~\cite{Simpson84Peano}, not reproduced in the monograph
\cite{simpson2009subsystems}. While the Big Five picture is known to
be incomplete \cite{SeetapunS95,Liu12}, the rest of the reverse
mathematics zoo concerns mostly combinatorial principles; ODEs, by
contrast, are search problems over the continuum, where the Big Five
retain a natural algorithmic interpretation. We extend this analysis
across all five levels and interpret each algorithmically as a
distinct form of search, yielding a representation-independent
classification indexed by regularity. 

In the computable analysis tradiction, the Weihrauch lattice classifies the uniform computational content of
mathematical problems \cite{brattka2021weihrauch}.
While powerful for fine-grained uniform reductions, Weihrauch complexity
is sensitive to representations and interfaces. Our reverse-mathematical
approach is complementary: it isolates the logical principles required for
solvability, 
connects naturally to implicit complexity and recursion schemas, abstracting from the fine-grained aspects of representations of this lattice.

\subsection{Our contributions per level.}

We present one ODE-related result for each level of the Big Five.
Some of these results are adaptations, with new proofs, of known
statements; others are, to the best of our knowledge, nowhere in the
literature.

\begin{itemize}
\item \emph{$\RCA_0$ (Cauchy-Lipschitz).} The classification of
Cauchy-Lipschitz in $\RCA_0$ is due to Simpson \cite{Simpson84Peano}.
We derive it as a corollary of our Cauchy-Osgood result. The proof
is entirely constructive, and avoids any compactness argument by
exploiting the Lipschitz hypothesis in \emph{both} variables rather
than only in the second one.

\item \emph{$\WKL_0$ (Cauchy-Osgood).} We present an
existence and uniqueness theorem under Osgood conditions. The
uniqueness part is classical; the existence part we have not found in
the literature. Our proof is tailored so that a single step requires
a $\WKL_0$-compactness argument: the Osgood modulus lets us replace
all other compactness calls by constructive arguments, tightening
Simpson's proof for Peano. A side benefit is that the proof
specialises cleanly to $\RCA_0$ when the modulus is Lipschitz, giving
the corollary above.

\item \emph{$\ACA_0$ (maximal Cauchy-Peano).} We show a maximal
version of the classical Cauchy-Peano existence theorem: any solution
we construct cannot be extended---to a larger domain---while remaining
a solution of the IVP. The original local version requires
$\WKL_0$ to construct one solution that may or may not extend; our
maximal version additionally requires multiple calls to the
Bolzano-Weierstra\ss{} theorem (equivalent to $\ACA_0$) to repeatedly
extract convergent subsequences from sequences in $\mathbb{R}^d$.
This result is, as far as we know, not available in the usual
literature.

\item \emph{$\mathrm{ATR}_0$ and $\Pi^1_1\text{-}\mathrm{CA}_0$
(Cauchy-Bournez-Gozzi).} For the last two levels we change the
nature of the result: from existence or uniqueness to
\emph{definability} of the solution, assuming such a solution exists.
We rely on a theorem of Bournez and Gozzi
\cite{gozzi2024STACSJournal,StacsBournezGozzi2024} that analytically
defines the solution via a transfinite recursion, originally motivated
by embedding transfinite Turing machine computations in ODEs. The
steps of the proof are theirs; our contribution is the
reverse-mathematical analysis of each transfinite step.

At level $\mathrm{ATR}_0$, we assume that the ordinal at which the
transfinite recursion stops is known a priori. This bound keeps the
logical strength exactly at $\mathrm{ATR}_0$. At level
$\Pi^1_1\text{-}\mathrm{CA}_0$, no a priori bound is available, and
the process may descend as deep as any countable ordinal, requiring
hyperarithmetical comprehension.
\end{itemize}

\subsection{The ODE theorems.}
We fix a rectangle
$\mathcal{R}=I\times\mathcal{D}\subset\mathbb{R}\times\mathbb{R}^n$
containing~$(t_0,y_0)$, to which $f$ is restricted throughout. We
refer to the associated problem~\eqref{eq:ivp} as \emph{the IVP}.

\begin{definition}
A \emph{local solution} of the IVP is a map
$y:[t_0,t_0+\epsilon]\to\mathbb{R}^d$, for some rational
$\epsilon>0$, satisfying the initial condition $y(t_0)=y_0$ and the
equation $y'(t)=f(t,y(t))$ on $[t_0,t_0+\epsilon]$. A \emph{maximal
solution} is a solution $y:J\to\mathbb{R}^d$ with $J\ni t_0$ admitting
no extension to a strictly larger interval that still solves the IVP.
\end{definition}

A prototypical existence theorem reads:
\begin{center}
\itshape Let $f:\mathcal{R}\to\mathbb{R}^d$ satisfy some prescribed
regularity property. Then the IVP has a solution.
\end{center}
The logical strength required to prove such a statement depends on the
regularity of $f$, and more precisely on the quantitative information
given about it: the mere existence of an object does not entail that
it can be recursively constructed.

We focus on the following five theorems, one per level of the Big
Five. The first three are classical existence (and uniqueness) results
with increasingly weak regularity.

\begin{theorem}[Cauchy-Lipschitz]\label{CLip}
If $f:\mathcal{R}\to\mathbb{R}^d$ is Lipschitz in both variables, the
IVP has a unique local solution.
\end{theorem}

Cauchy-Lipschitz, also known as Picard-Lindel\"of, is the
textbook existence-and-uniqueness result. Our formulation is slightly
stronger than the classical one (which requires $f$ continuous and
Lipschitz in its second variable only), and avoids any use of
compactness; see Section~\ref{sec:populating}.

\begin{theorem}[Cauchy-Osgood]\label{COsgood}
Let $f:\mathcal{R}\to\mathbb{R}^d$ be continuous and
$\omega:[0,\delta)\to[0,\infty)$ continuous and increasing with
$\omega(0)=0$, satisfying
$$\|f(t,y)-f(t,z)\|\leq\omega(\|y-z\|)$$ for all
$(t,y),(t,z)\in\mathcal{R}$ with $\|y-z\|\leq\delta$, and the Osgood
blow-up condition $$\int_0^\epsilon dr/\omega(r)=+\infty$$ for all
$\epsilon>0$. Then the IVP has a unique local solution.
\end{theorem}

The Osgood modulus $\omega$ generalises the Lipschitz condition (which
corresponds to $\omega(r)=Lr$) and its blow-up condition is the
minimal requirement that still forces uniqueness. The uniqueness part
is the classical result; the existence part, to the best of our
knowledge, is new. Cauchy-Lipschitz becomes a direct corollary of
Cauchy-Osgood.

\begin{theorem}[Cauchy-Peano, maximal]\label{CPeano}
If $f:\mathcal{R}\to\mathbb{R}^d$ is continuous, the IVP has a maximal
solution.
\end{theorem}

The classical Cauchy-Peano theorem only asserts the existence of a
\emph{local} solution on some interval determined by
$\sup_{\mathcal{R}}\|f\|$; our maximal version, which guarantees
that the domain of the solution cannot be strictly extended, is, to
our knowledge, new.

The last two theorems differ in spirit. They do not assert existence;
they assert that a given solution, known to be unique, is
\emph{definable} in a given subsystem, provided the vector field
satisfies an additional solvability hypothesis.

\begin{theorem}[Cauchy-Bournez-Gozzi, bounded]\label{CBG1}
If the IVP has a unique solution and $f$ is $\alpha$-solvable (see
Section~\ref{sec:populating}), then the solution is definable in
$\mathrm{ATR}_0$.
\end{theorem}

\begin{theorem}[Cauchy-Bournez-Gozzi, unbounded]\label{CBG2}
If the IVP has a unique solution and $f$ is solvable (see
Section~\ref{sec:populating}), then the solution is definable in
$\Pi^1_1\text{-}\mathrm{CA}_0$.
\end{theorem}

The distinction between these two theorems lies in whether an ordinal
bound on the transfinite recursion defining the solution is known a
priori ($\mathrm{ATR}_0$) or not ($\Pi^1_1\text{-}\mathrm{CA}_0$). In
both cases, the proof strategy follows Bournez and Gozzi
\cite{gozzi2024STACSJournal,StacsBournezGozzi2024}; our contribution is
the reverse-mathematical calibration of each transfinite step.


\subsection{Ingredients from reverse mathematics.}
All proofs rest on the following equivalences and consequences
\cite{simpson2009subsystems}.

\begin{theorem}[{\cite[IV.2.3]{simpson2009subsystems}}]\label{WKLassertions}
The following are equivalent over RCA\({}_0\).
\begin{enumerate}
\item Weak K\"oning's Lemma.
\item Every continuous function on the closed interval \(0\leq x\leq1\) is uniformly continuous.
\item Every bounded, uniformly continuous function on \(0\leq x\leq1\) has a supremum.
\end{enumerate}
\end{theorem}

\begin{lemma}[Corollary]\label{coroWKL}
The following are provable in WKL\({}_0\) over RCA\({}_0\)
\begin{enumerate}
\item\label{magic M} If \(f:\mathcal{R}\to\mathbb{R}^d\) is continuous. Then there exists a rational \(M>0\) such that for all \((t,y)\in\mathcal{R}\) we have \(\norm{f(t,y)}\leq M\).
\item\label{tetoile} Every non-empty closed subset of \([0,1]\) has a supremum.
\end{enumerate}
\end{lemma}

\begin{theorem}[{\cite[IV.8.2]{simpson2009subsystems}}]\label{WKLequiv}
The following assertions are equivalent over RCA\({}_0\).
\begin{enumerate}
\item WKL\({}_0\).
\item\label{WKLexists} If \(f(x,y)\) is continuous and has a modulus of uniform continuity in some neighbourhood of \(x=0,\ y=0\), then the IVP has a continuously differentiable solution in some interval containing \(x=0\).
\end{enumerate}
\end{theorem}

\begin{theorem}[{\cite[III.2.2]{simpson2009subsystems}}]\label{BWthm}
The following assertions are equivalent over RCA\({}_0\).
\begin{enumerate}
\item ACA\({}_0\).
\item The Bolzano-Weierstra{\ss} theorem: every bounded sequence of real numbers contains a convergent subsequence.
\end{enumerate}
\end{theorem}

%
\section{Populating the Big Five}

\subsection{Cauchy-Lipschitz in RCA\({}_0\)}

The hypothesis of \cref{CLip} is a stronger version of the usual hypothesis for the Cauchy-Lipschitz theorem, also known as the Picard-Lindel\"of theorem. In the usual version, the function \(f\) is required to be continuous and Lipschitz only in its second variable. Since all the proofs (to our knowledge) are based on the claim ``every continuous function over a compact is uniformly continuous''---\cref{WKLassertions}, assertion (2)---in order to find \(M\) as in \cref{magic M} in \cref{coroWKL}, it seems that the ``standard'' hypothesis forces to resort to compactness arguments. The stronger hypothesis fixes this situation since \(f\) being Lipschitz in its two variables amounts to have a bound over time and (not only) space, which immediately yields the aforementioned \(M\). The rest of the proof can be completely handled within RCA\({}_0\), \cite[Theorem IV.8.3]{simpson2009subsystems}{---similar to our version, but observing that Simpson's result considers \(M\) as an input.} It can also be stated as a corollary of the Cauchy-Osgood result, since the only WKL\({}_0\)-step is avoided by the Lipschitz condition. This amounts to say.

\begin{theorem}\label{CLRCA}
\cref{CLip} is provable in RCA\({}_0\).
\end{theorem}

\subsection{Cauchy-Osgood in WKL\({}_0\)}

The usual ``Osgood theorem'' is in fact a uniqueness theorem; the statement does not consider \(f\) to be continuous and the existence of a solution is not guaranteed. Usual proofs compare two solutions in a similar fashion as ours (in the uniqueness part), but our endeavours require a more subtle treatment of each step. However, and to the best of our knowledge, there is no proof of existence under Osgood conditions (plus continuity) in the literature. A quick proof could simply invoke Cauchy-Peano's theorem, in WKL\({}_0\), but we found ourselves in the need to go extra lengths to reduce the calls of a WKL\({}_0\) to just one. This (logical) simplification of the argument allows us then to have \cref{CLRCA} as a corollary almost directly.

The hypothesis in \cref{COsgood} can be seen as a generalization of the Lipschitz conditions. Instead of requiring a linear modulus, a simply continuous and non-decreasing modulus is required, plus a blow-up condition that is clearly satisfied by a Lipschitz function. This tightening and the lack of control over the time variable force the need for a compactness argument, which increases the logical power required.

\begin{theoremrep}
\cref{COsgood} is equivalent to WKL\({}_0\) over RCA\({}_0\).
\end{theoremrep}
\begin{proof}

\((\Rightarrow)\)We first show that WKL\({}_0\) implies \cref{COsgood}.\newline
\textbf{Existence:}   By \cref{coroWKL}(i), (this is a  WKL\({}_0\) statement) 
there exists \(M>0\) such that \(\lVert{f(t,y)-f(t',y')}\rVert\leq M\) on \(\mathcal{R}\).  We reason in RCA\({}_0\).

Let \(r>0\) be such that \(\overline{\mathcal{B}(y_0,r)}\subseteq\mathcal{D}\). Choose \(t^*>0\) such that \(I:=[t_0,t_0+t^*]\subseteq\mathcal{I}\) and \(Mt^*\leq r\). For \(k\in\mathbb{N}\), let \(t_i=t_0+it^*/k\), for \(0\leq i\leq k\), and for \(t\in[t_i,t_{i+1}]\) define \(y_k:I\to\mathbb{R}^n\) by
\[y_k(t):=y_k(t_i)+(t-t_i)f(t_i,y_k(t_i)),\text{ with }y_k(t_0)=y_0.\]
In this way, \(y_k:I\to\mathbb{R}^n\) is a piece-wise map with slope sampling taken in the left endpoint of \([t_i,t_{i+1}]\) for each \(0\leq i\leq k\). We now show that the sequence \((y_k)_k\) is a uniform Cauchy sequence, that is, \(\sup_{t\in I}\norm{y_k(t)-y_l(t)}\to0\) as \(k,l\to\infty\).

For \(t\in[t_i,t_{i+1}]\), we have
\begin{align*}
\norm{y_k(t)-y_0}&\leq\sum_{j<i}\norm{y_k(t_{j+1})-y_k(t_j)}+\norm{y_k(t)-y_k(t_i)}\\
&\leq M\left(i\frac{t^*}{k}+(t-t_i)\right)=(t-t_0)M \leq Mt^*\leq r,
\end{align*}
so \(y_k(I)\subseteq\overline{\mathcal{B}(y_0,r)}\).

Let \(f_k(t):=f(t_i,y_k(t_i))\) on \([t_i,t_{i+1}]\). Then
\[y_k(t)=y_0+\int_{t_0}^tf_k(s)ds.\]

Fix \(t\in I\) and, for each \(k\), let \(i=i(t,k)\) be the unique index with \(t\in[t_i,t_{i+1}]\). Set \(F(t,k)=\lVert f(t,y_k(t))-f_k(t)\rVert\). By uniform continuity of \(f\) on \(\mathcal{R}\), we have
\begin{align*}
F(t,k)&\leq\omega\left(\lVert y_k(t)-y_k(t_i)\rVert\right)+\epsilon_k 
\leq\omega\left(\frac{Mt^*}{k}\right)+\epsilon_k,
\end{align*}
where \(\epsilon_k=\sup_{|t-t_i|\leq t^*/k}\lVert f(t,y_k(t_i))-f(t_i,y_k(t_i))\rVert\to0\) as the mesh size \(t^*/k\) goes to \(0\). Hence
\[\norm{f(t,y_k(t))-f_k(t)}\to0\text{ as }k\to\infty.\]

Let \(k,l\) be two distinct positive integers, then
\[y_k(t)-y_l(t)=\int_{t_0}^tf_k(s)-f_l(s)ds.\]

Set \(S(t)=\sup_{t'\leq t}\lVert y_k(t')-y_l(t')\rVert\). Write \(t_j^{(k)}=t_0+jt^*/k\) and \(t_j^{(l)}=t_0+jt^*/l\) for the meshes underlying \(y_k\) and \(y_l\). For \(s\in I\), denoting by \(t_{i(s,k)}^{(k)}\) (resp.\ \(t_{i(s,l)}^{(l)}\)) the left endpoint of the subinterval of \(y_k\) (resp.\ \(y_l\)) containing \(s\), we have
\[\lVert f_k(s)-f_l(s)\rVert\leq\epsilon_k+\epsilon_l+\omega(Mt^*/k)+\omega(Mt^*/l)+\omega(S(s)),\]
where \(E_{k,l}=\epsilon_k+\omega(Mt^*/k)+\epsilon_l+\omega(Mt^*/l)\to0\) as \(k,l\to\infty\). Thus
\[S(t)\leq(t-t_0)E_{k,l}+\int_{t_0}^t\omega(S(s))ds.\]
Let \((\tau_j)_{0\leq j\leq N}\) be the common refinement of \((t_j^{(k)})\) and \((t_j^{(l)})\), so that on each \([\tau_j,\tau_{j+1}]\) both \(f_k\) and \(f_l\) are constant. Set \(S_j:=S(\tau_j)\). Using the monotonicity of \(S\) we obtain
\[S_{j+1}-S_j\leq(\tau_{j+1}-\tau_j)\bigl(E_{k,l}+\omega(S_j)\bigr).\]
For \(k,l\) large enough so that \(E_{k,l}\) is negligible compared with \(\omega(S_j)\),
\[\frac{S_{j+1}-S_j}{\omega(S_j)}\leq 2(\tau_{j+1}-\tau_j).\]

Integrating and summing over \(j\) we obtain,
\[\int_0^{S(t)}\frac{dr}{\omega(r)}\leq2(t-t_0).\]
This combined with the Osgood condition forces \(S(t)=0\). Therefore \((y_k)\) is uniformly Cauchy. Since \(C(I,\mathbb{R}^n)\) is complete, \(y_k\to y\). Passing to the limit in
\[y_k(t)=y_0+\int_{t_0}^tf_k(s)ds,\text{ yields }y(t)=y_0+\int_{t_0}^tf(s,y(s))ds,\]
So \(y\) solves the IVP.

\textbf{Uniqueness:} 
We argue in RCA\({}_0\). Let \(y,z\) be two solution and set \(u(t)=\lVert y(t)-z(t)\rVert\). Then
\[\begin{aligned}
|u'(t)| &\leq \lVert f(t,y(t))-f(t,z(t))\rVert\\
&\leq \omega(u(t)).
\end{aligned}\]
Assume \(u(t_1)>0\) for some \(t_1>t_0\). By \cref{tetoile} in \cref{coroWKL}, there exists \(t^*\in(t_0,t_1)\) such that \(u=0\) on \((t_0,t^*)\) and \(u>0\) on \((t^*,t_1)\). Then for \(t\in(t^*,t_1)\),
\[\frac{u'(t)}{\omega(u(t))}\leq1
\Rightarrow\int_0^{u(t_1)}\frac{dr}{\omega(r)}\leq t_1-t^*<\infty,\]
contradicting the Osgood condition. Hence, \(u\equiv0\) and the solution is unique.

\((\Leftarrow)\) This direction follows immediately from \cref{WKLequiv}, since existence in Cauchy-Osgood implies \cref{WKLexists}.
\end{proof}

\subsection{Cauchy-Peano (maximal) in ACA\({}_0\)}

The Cauchy-Peano existence theorem states, typically, that if \(f\) is continuous, then there exists a local solution. This solution lives in an interval \(J\subseteq I\) determined by \(\sup_\mathcal{R}\norm{f(t,y)}\). As shown by \cite[Theorem IV.8.1]{simpson2009subsystems}, this theorem is provable in WKL\({}_0\) and by \cref{WKLequiv}, we obtain the missing direction to conclude they are equivalent. We present a maximal version, in the sense that the solution \(y:J^*\to\mathbb{R}^d\) is a maximal solution, that is, any extension of \(y\) beyond \(J^*\) is no longer a solution, so the interval where the solution is defined is as big as it gets.

\paragraph*{A brief note on Simpson's proof:} A classical maths proof of the Cauchy-Peano theorem involves the Arzel\`a-Ascoli theorem (Simpson calls it the Ascoli lemma): every bounded equicontinuous sequence of functions has a uniformly convergent subsequence. Over RCA\({}_0\), Ascoli lemma is equivalent to ACA\({}_0\) and Cauchy-Peano is equivalent to WKL\({}_0\)---both results can be found in \cite{Simpson84Peano}. The straightforward guess is that the WKL\({}_0\) proof requires more work that the ACA\({}_0\) proof, and it sure does. In Simpson's words:
\begin{center}
\textit{``...WKL\({}_0\) gives a good theory of continuity, while ACA\({}_0\) gives in addition a good theory of sequential convergence.''}
\end{center}
The sequential argument of Ascoli lemma must be replaced, in the case of Simpson's proof, by a Heine/Borel style covering argument.

\FINAL{
\todo[inline]{Can we argue that the proof of the max version cannot be done in WKL\({}_0\)?}
}

\begin{theoremrep}
\cref{CPeano} is provable in ACA\({}_0\) over RCA\({}_0\).
\end{theoremrep}

\begin{proof}
For the existence of a solution, we invoke the usual Cauchy-Peano statement, that is, if \(f\) is continuous, then there exist \(\epsilon>0\) and a (local) solution \(\phi:[t_0,t_0+\epsilon]\to\mathbb{R}^d\) to the IVP. This argument is equivalent to WKL\({}_0\), (\cite{Simpson84Peano}), so each call for a local solution requires WKL\({}_0\).

We reason over RCA\({}_0\) and we indicate where more logical strength is required. Define
\[T:=\{t>t_0\mid\text{ there exists a solution in }[t_0,t]\}\text{  and  }\beta=\sup_I T.\]
These definitions require WKL\({}_0\). Observe that \(T\neq\emptyset\) since \(t_0+\epsilon\in T\). Take an increasing rational sequence \((q_n)_n\) in \([t_0,\beta]\) such that \(q_n\to\beta\). For each \(n\), choose a solution \(\phi_n:[t_0,q_n]\to\mathbb{R}^d\). All of them satisfy
\[\phi_n(t)=y_0+\int_{t_0}^tf(s,\phi_n(s))ds.\]
Hence,
\[\norm{\phi_n(t)-\phi_m(t)}\leq M(t-t_0),\]
with \[M\geq\norm{f}_\infty,\] which exists because \(f\) is continuous over the compact domain \(\mathcal{R}\)---the existence of such an \(M\) requires WKL\({}_0\). This implies that all of the \(\phi_n\)s are uniformly bounded, uniformly Lipschitz, and share a common modulus of uniform continuity.

Fix \(b<\beta\) and enumerate a dense sequence \((t_i)_i\) in \([t_0,b]\). Since \((\phi_n(t_1))_n\) is a bounded sequence in \(\mathbb{R}^d\), one application of the Bolzano-Weierstra{\ss} theorem yields a convergent subsequence, that is, \(k_1:\mathbb{N}\to\mathbb{N}\) such that \(\phi_{k_1(n)}(t_1)\) converges. Write\(\phi^{(1)}_n=\phi_{k_1(n)}\) and observe that \((\phi^{(1)}_n(t_2))_n\) is a bounded sequence of \(\mathbb{R}^d\), so, another call to Bolzano-Weierstra{\ss} theorem gives us \(k_2:\mathbb{N}\to\mathbb{N}\) such that \(\phi^{(1)}_{k_2(n)}(t_2)\) converges. Write \(\phi^{(2)}_n=\phi^{(1)}_{k_2(n)}\) and observe that the sequence \((\phi^{(2)}_n)_n\) coverges at \(t_1\) and at \(t_2\). Repeating this, recursively, at stage \(m\) we have \(\phi^{(m-1)}_n\) converges at \(t_1,\ldots,t_{m-1}\) and we can extract a convergent subsequence \((\phi^{(m)}_n)_n\) that converges at \(t_m\). Over RCA\({}_0\), the Bolzano-Weierstra{\ss} theorem is equivalent to ACA\({}_0\), so this entire recursive argument requires ACA\({}_0\).

Define the diagonal sequence \(y_n:=\phi^{(n)}_n\) and observe that \(y_n(t_m)\) converges for all \(m\). Since we have uniform Lipschitz bounds, we can extend the convergence from the dense set \(\{t_i\mid i\}\) to all \(t\in[t_0,b]\). Hence, \(y_n\) converges uniformly on \([t_0,b]\), and since limits of solutions are solutions (Fundamental Theorem of Calculus for continuous functions, available in WKL\({}_0\), see \cite[Theorem IV.2.7]{simpson2009subsystems} ) we have that the limit function is a solution over \([t_0,b]\). Since \(b<\beta\) was arbitrary, we conclude that there exists a solution \(y:[t_0,\beta)\to\mathbb{R}^d\) to the IVP.

Since the slope of the solution is bounded the solution stays bounded as \(t\to\beta\). If \(\beta\) is finite, continuity gives a limit point at \(\beta\), so the Cauchy-Peano (local) existence theorem allows us to extend a solution beyond \(\beta\), which is a contradiction. Hence \(\beta=\sup I\).
\end{proof}

\subsection{Cauchy-Bournez-Gozzi from ATR\({}_0\) to \(\Pi_1^1\)-CA\({}_0\)}\label{sec:populating}

We now turn from existence and uniqueness to definability: the logical strength required to name a mathematical object. \(\ACA_0\) already provides comprehension for arithmetical formulas and covers \cref{CLip,COsgood,CPeano}. Beyond it, \(\mathrm{ATR}_0\) permits definitions by transfinite recursion along well-orders, and \(\Pi^1_1\)-CA\({}_0\) extends comprehension to analytic formulas. \cref{CBG1,CBG2} inhabit these two upper levels.

For convenience of our exposition, we start from the end. Let \(f:\mathcal{R}\to\mathbb{R}^d\) be a Baire-one function, that is, \(f\) is the point-wise limit of continuous functions. For a closed set \(K\subseteq\mathcal{R}\), denote by \(D_K\) the set of points at which \(f|_K\) is discontinuous, where \(f|_{K}\) is the restriction of \(f\) to \(K\). Consider the following transfinite recursion:
\[\mathcal{R}^{(0)}=\mathcal{R},\quad\mathcal{R}^{(\alpha+1)}=D_{\mathcal{R}^{(\alpha)}},\quad\text{and}\quad\mathcal{R}^{(\lambda)}=\cap_{\beta<\lambda}\mathcal{R}^{(\beta)},\]
where \(\alpha\) is a non-limit ordinal and \(\lambda\) is a limit ordinal.
\begin{definition}[Adapted from \cite{StacsBournezGozzi2024}]\label{solv}
A Baire-one function \(f:\mathcal{R}\to\mathbb{R}^d\) is called \textbf{solvable} if there exists a countable ordinal \(\alpha\) such that \(\mathcal{R}^{(\alpha-1)}\neq\emptyset\) and \(\mathcal{R}^{(\alpha)}=\emptyset\). If the countable ordinal \(\alpha\) is known (beforehand), then \(f\) is called \textbf{\(\alpha\)-solvable}.
\end{definition}

For an \(\alpha\)-solvable function, we consider the ordinal \(\alpha\) as part of the date given by the function. This  tells us the  level  at which there are no discontinuity points left. This extra piece of information is precisely what separates ATR\({}_0\) from \(\Pi^{1}_{1}\)-CA\({}_0\)---bounded vs unbounded transfinite iteration of ACA\({}_0\) statements.

\begin{theorem}\label{CBalpha}
\cref{CBG1} is provable in ATR\({}_0\).
\end{theorem}

\begin{proof}[Sketch of the proof]

We begin by observing that the set \(D_\mathcal{R}=\mathcal{R}^{(0)}\) is a countable union of nowhere dense sets (\cite[Lemma 2]{StacsBournezGozzi2024}). This is a RCA\({}_0\) as its proof relies on the Baire Category theorem that it is known to be provable in RCA\({}_0\).

A cardinality argument shows that the projection \(J^*_0\) of \(\mathcal{R}^{(0)}\setminus\mathcal{R}^{(1)}\) over \(\mathcal{I}\) is, at most, the union of \(\beta\) disjoint intervals, for some countable ordinal \(\beta<\omega_1\). Let \(J\) be one such interval and observe that \(f|_J\) is continuous. A single call to \cref{CPeano} yields a maximal solution in \(J\). Proceeding in the same way for each of the \(\beta\) intervals we obtain---by taking unions---a solution \(y_0:J^*_0\to\mathbb{R}^d\). In general, for \(\gamma\leq\alpha\) we consider the projection \(J^*_\gamma\) of \(\mathcal{R}^{(\gamma)}\setminus\mathcal{R}^{(\gamma+1)}\) over \(\mathcal{I}\) and proceed in the same fashion over each of the related disjoint intervals to obtain a solution \(y_\gamma:J^*_\gamma\to\mathbb{R}^d\). Since \(f\) is \(\alpha\)-solvable, \(\mathcal{R}^{(\alpha)}=\emptyset\) and hence \(\mathcal{R}^{(0)}\setminus\mathcal{R}^{(\alpha)}=\mathcal{R}\). We conclude by defining a solution \(y\) over \(\cup_{\eta\leq\alpha-1}J^*_\eta=\mathcal{I}\) by \(y|_{J^*_\eta}=y_\eta\) for each \(\eta\leq\alpha-1\).

Since \cref{CPeano} is an ACA\({}_0\) statement, we are using an \(\alpha\)-deep transfinite recursion, that is, we require ATR\({}_0\) to proceed as we did.
\end{proof}

For solvable functions where the ordinal \(\alpha\) is not known, we have the following.

\begin{theorem}
\cref{CBG2} is provable in \(\Pi_1^1\)-CA\({}_0\).
\end{theorem}
\begin{proof}[Sketch of the proof]
We proceed exactly as in the sketch of the proof of \cref{CBalpha} with the sole difference that \(\alpha\) is not known beforehand, the transfinite recursion can extend up to arbitrarily large countable ordinals. This is precisely what \(\Pi^1_1\)-CA\({}_0\) provides.
\end{proof}

This completes the population of the Big Five by natural ODE statements indexed by the regularity of \(f\), with ATR\({}_0\) marking a shift from existence to the vocabulary needed to \emph{name} solutions.

\section{Conclusion and Perspectives}
\label{sec:conclusion}
\bigskip

We have shown that every level of the Big Five is inhabited by a
natural ODE theorem, with exact equivalences between regularity
assumptions on the vector field~$f$ and logical principles:
Cauchy-Lipschitz at~$\RCA_0$, Cauchy-Osgood at~$\WKL_0$, maximal
Cauchy-Peano at~$\ACA_0$, and Bournez-Gozzi-type definability at~$\mathrm{ATR}_0$
and $\Pi^1_1\text{-}\mathrm{CA}_0$. The regularity of~$f$ emerges as an
intrinsic algorithmic invariant determining which logical principles
any solver must invoke.

\medskip
\paragraph{Intrinsic explanations for computer algebra.} The discussion in
Section~\ref{sec:brm-eff-odes} about analytic-continuation 
illustrates what ``intrinsic'' can mean in practice. Symbolic methods
solve polynomial ODEs over~$\mathbb{Q}[[X]]$ in quasi-linear time; yet
evaluating the real solution on~$[0,T]$ is inherently sequential in
the number of continuation patches, a quantity that depends on the
geometry of the solution rather than on any implementation choice. This seems intuitive, but,
can we prove it cannot be done better?

Precisely, the proof-theoretic classification makes this obstruction structural.
For each regularity class, a uniformly faster solver would entail the
violation of an established result of complexity or computability
theory: at the polynomial level, the expressibility of $E(x)=2^x$ as a
bounded recursion on notation; at the $\WKL_0$ level, computable
solutions for every computable continuous right-hand side,
contradicting~\cite{PR79,Abe71}; at higher levels, the avoidance of
comprehension principles (arithmetical, transfinite,
hyperarithmetical) whose necessity is part of the equivalences we
prove. 

Dually, supplying auxiliary intrinsic data, such as length bounds,
radii of convergence, moduli of continuity, places the problem at a
lower stratum and restores tractability. The contribution is thus not
to diagnose what symbolic solvers have missed, but to pin each barrier
down to the precise, proof-theoretically traceable assumption of
complexity or computability theory that any improvement would need to
overturn. 

\medskip
\paragraph{Reverse mathematics as a representation-invariant toolkit.}
Our broader programme is to exploit reverse mathematics as a common
language for discussing algorithmic difficulty across communities
whose native encodings are hard to compare directly. The five
equivalences in this paper are a proof of concept in the context
of initial value problems. The same methodology should apply wherever
representation choices obscure the underlying computational content:
partial and delay differential equations; differential-algebraic
systems as handled in symbolic computation; fixed-point and
equilibrium problems from descriptive complexity; and, through
bounded reverse mathematics, a finer stratification within the
efficient regime that would match the complexity classes actually
relevant to computer algebra. In each case, the goal is the same: to
trade comparisons of encodings for comparisons of logical content, and
thereby recover the structural reasons behind algorithmic difficulty.

\newpage

\bibliographystyle{plain}
\bibliography{reference-biblio.bib}

@article{JOC2007,
	_bib2doi_confirmed = {true},
	_bib2doi_finished = {true},
	_bib2doi_selected = {dblp:/rec/journals/jc/BournezCGH07.bib},
	author = {Bournez, Olivier and Campagnolo, Manuel Lameiras and Gra{\c{c}}a, Daniel Silva and Hainry, Emmanuel},
	bibsource = {dblp computer science bibliography, https://dblp.org},
	biburl = {https://dblp.org/rec/journals/jc/BournezCGH07.bib},
	date-added = {2026-01-31 15:40:39 +0100},
	date-modified = {2026-01-31 15:40:39 +0100},
	doi = {10.1016/j.jco.2006.12.005},
	journal = {Journal of Complexity},
	month = {June},
	number = {3},
	pages = {317--335},
	pdf = {https://www.lix.polytechnique.fr/\textasciitilde bournez/uploads/Main/JournalOfComplexity2006.pdf},
	timestamp = {Tue, 16 Feb 2021 00:00:00 +0100},
	title = {Polynomial differential equations compute all real computable functions on computable compact intervals},
	topics = {Analog},
	url = {http://dx.doi.org/10.1016/j.jco.2006.12.005},
	volume = {23},
	year = {2007},
	bdsk-url-1 = {http://dx.doi.org/10.1016/j.jco.2006.12.005}}

@incollection{CIEChapter2007,
	_bib2doi_finished = {true},
	address = {New York},
	author = {Bournez, Olivier and Campagnolo, Manuel L.},
	title = {A Survey on Continuous Time Computations},
	date-added = {2026-01-31 15:40:39 +0100},
	date-modified = {2026-01-31 15:40:39 +0100},
	editor = {Cooper, S.B. and L{\"o}we, B. and Sorbi, A.},
	pages = {383--423},
	pdf = {https://www.lix.polytechnique.fr/\textasciitilde bournez/uploads/Main/SurveyContinuousTime.pdf},
	publisher = {Springer-Verlag},
	booktitle = {New Computational Paradigms. Changing Conceptions of What is Computable},
	topics = {Analog, ModelsSurvey},
	year = {2008}}

@inproceedings{bournez2012complexity,
	_bib2doi_confirmed = {true},
	_bib2doi_finished = {true},
	_bib2doi_selected = {dblp:/rec/conf/issac/BournezGP12.bib},
	author = {Bournez, Olivier and Gra{\c{c}}a, Daniel Silva and Pouly, Amaury},
	bibsource = {dblp computer science bibliography, https://dblp.org},
	biburl = {https://dblp.org/rec/conf/issac/BournezGP12.bib},
	booktitle = {International Symposium on Symbolic and Algebraic Computation, ISSAC'12, Grenoble, France - July 22 - 25, 2012},
	date-added = {2026-01-31 15:40:39 +0100},
	date-modified = {2026-01-31 15:40:39 +0100},
	doi = {10.1145/2442829.2442849},
	editor = {van der Hoeven, Joris and van Hoeij, Mark},
	pages = {115--121},
	publisher = {{ACM}},
	timestamp = {Wed, 29 Jul 2020 01:00:00 +0200},
	title = {On the complexity of solving initial value problems},
	topics = {Analog,Dynamical,GPAC},
	url = {https://doi.org/10.1145/2442829.2442849},
	year = {2012},
	bdsk-url-1 = {https://doi.org/10.1145/2442829.2442849}}

@article{JournalACM2017,
	_bib2doi_confirmed = {true},
	_bib2doi_finished = {true},
	_bib2doi_selected = {dblp:/rec/journals/jacm/BournezGP17.bib},
	author = {Bournez, Olivier and Gra{\c{c}}a, Daniel Silva and Pouly, Amaury},
	bibsource = {dblp computer science bibliography, https://dblp.org},
	biburl = {https://dblp.org/rec/journals/jacm/BournezGP17.bib},
	date-added = {2026-01-31 15:40:39 +0100},
	date-modified = {2026-01-31 15:40:39 +0100},
	doi = {10.1145/3127496},
	journal = {Journal of the ACM},
	keywords = {Computer Science - Computational Complexity},
	number = {6},
	pages = {38:1--38:76},
	timestamp = {Tue, 07 May 2024 01:00:00 +0200},
	title = {Polynomial Time Corresponds to Solutions of Polynomial Ordinary Differential Equations of Polynomial Length},
	topics = {Analog,Dynamical,GPAC},
	volume = {64},
	year = {2017},
	bdsk-url-1 = {http://doi.acm.org/10.1145/3127496},
	bdsk-url-2 = {https://doi.org/10.1145/3127496}}

@incollection{bournez2021survey,
	_bib2doi_finished = {true},
	anciencode = {bournez2021survey},
	author = {Bournez, Olivier and Pouly, Amaury},
	booktitle = {Handbook of Computability and Complexity in Analysis},
	date-added = {2026-01-31 15:40:39 +0100},
	date-modified = {2026-01-31 15:40:39 +0100},
	doi = {10.1007/978-3-030-59234-9\_6},
	editor = {Brattka, Vasco and Hertling, Peter},
	publisher = {Springer},
	title = {A Survey on Analog Models of Computation},
	topics = {Analog,Dynamical,GPAC,ModelsSurvey},
	year = {2021},
	bdsk-url-1 = {https://doi.org/10.1007/978-3-030-59234-9%5C_6}}

@article{MFCSJournal,
	_bib2doi_confirmed = {true},
	_bib2doi_finished = {true},
	_bib2doi_selected = {dblp:/rec/journals/cc/BournezD23.bib},
	author = {Bournez, Olivier and Durand, Arnaud},
	bibsource = {dblp computer science bibliography, https://dblp.org},
	biburl = {https://dblp.org/rec/journals/cc/BournezD23.bib},
	date-added = {2026-01-31 15:40:39 +0100},
	date-modified = {2026-01-31 15:40:39 +0100},
	doi = {10.1007/s00037-023-00240-1},
	journal = {Computational Complexity},
	number = {2},
	pages = {7},
	publisher = {Springer},
	timestamp = {Tue, 07 May 2024 01:00:00 +0200},
	title = {A Characterization of Functions over the Integers Computable in Polynomial Time Using Discrete Ordinary Differential Equations},
	topics = {Analog,DiscreteODE},
	url = {https://doi.org/10.1007/s00037-023-00240-1},
	volume = {32},
	year = {2023},
	bdsk-url-1 = {https://doi.org/10.1007/s00037-023-00240-1}}

@inproceedings{StacsBournezGozzi2024,
	_bib2doi_confirmed = {true},
	_bib2doi_finished = {true},
	_bib2doi_selected = {dblp:/rec/conf/stacs/BournezG24.bib},
	author = {Bournez, Olivier and Gozzi, Riccardo},
	bibsource = {dblp computer science bibliography, https://dblp.org},
	biburl = {https://dblp.org/rec/conf/stacs/BournezG24.bib},
	booktitle = {Symposium on Theoretical Aspects of Computer Science (STACS), Clermont-Ferrand, France},
	date-added = {2026-01-31 15:40:39 +0100},
	date-modified = {2026-01-31 15:40:39 +0100},
	doi = {10.4230/LIPIcs.STACS.2024.20},
	timestamp = {Sun, 06 Oct 2024 01:00:00 +0200},
	title = {Solving discontinuous initial value problems with unique solutions is equivalent to computing over the transfinite},
	topics = {Analog,Dynamical},
	year = {2024},
	bdsk-url-1 = {https://doi.org/10.4230/LIPIcs.STACS.2024.20}}

@techreport{gozzi2024STACSJournal,
	_bib2doi_confirmed = {true},
	_bib2doi_finished = {true},
	_bib2doi_selected = {dblp:/rec/journals/corr/abs-2405-00165.bib},
	author = {Bournez, Olivier and Gozzi, Riccardo},
	bibsource = {dblp computer science bibliography, https://dblp.org},
	biburl = {https://dblp.org/rec/journals/corr/abs-2405-00165.bib},
	date-added = {2026-01-31 15:40:39 +0100},
	date-modified = {2026-01-31 15:40:39 +0100},
	doi = {10.48550/arXiv.2405.00165},
	eprint = {2405.00165},
	eprinttype = {arXiv},
	journal = {CoRR},
	timestamp = {Sun, 09 Jun 2024 01:00:00 +0200},
	title = {Solvable Initial Value Problems Ruled by Discontinuous Ordinary Differential Equations},
	topics = {Analog,Dynamical},
	url = {https://doi.org/10.48550/arXiv.2405.00165},
	volume = {abs/2405.00165},
	year = {2024},
	institution = {\'Ecole Polytechnique, France},
	bdsk-url-1 = {https://doi.org/10.48550/arXiv.2405.00165}}

@inproceedings{Icalp2024,
	_bib2doi_confirmed = {true},
	_bib2doi_finished = {true},
	_bib2doi_selected = {dblp:/rec/conf/icalp/BlancB24.bib},
	author = {Blanc, Manon and Bournez, Olivier},
	bibsource = {dblp computer science bibliography, https://dblp.org},
	biburl = {https://dblp.org/rec/conf/icalp/BlancB24.bib},
	booktitle = {51st International Colloquium on Automata, Languages, and Programming, {ICALP} 2024, Tallinn, Estonia, July 8-12, 2024},
	date-added = {2026-01-31 15:40:39 +0100},
	date-modified = {2026-01-31 15:40:39 +0100},
	doi = {10.4230/LIPIcs.ICALP.2024.129},
	editor = {Bringmann, Karl and Grohe, Martin and Puppis, Gabriele and Svensson, Ola},
	month = {July},
	pages = {129:1--129:22},
	publisher = {Schloss Dagstuhl - Leibniz-Zentrum f{\"{u}}r Informatik},
	remplacementadoi = {https://compose.ioc.ee/icalp2024/\#icalp},
	series = {LIPIcs},
	timestamp = {Sat, 15 Nov 2025 00:00:00 +0100},
	title = {The Complexity of Computing in Continuous Time: Space Complexity Is Precision},
	topics = {Analog,Dynamical,GPAC},
	url = {https://doi.org/10.4230/LIPIcs.ICALP.2024.129},
	volume = {297},
	year = {2024},
	bdsk-url-1 = {https://doi.org/10.4230/LIPIcs.ICALP.2024.129}}

@article{Computability2024Gozzi,
	_bib2doi_finished = {true},
	author = {Gozzi, Riccardo and Bournez, Olivier},
	date-added = {2026-01-31 15:40:39 +0100},
	date-modified = {2026-01-31 15:40:39 +0100},
	journal = {Computability},
	note = {Accepted for publication. To appear},
	title = {Set Descriptive Complexity of Solvable Functions},
	topics = {Analog,Dynamical},
	year = {2025},
	bdsk-url-1 = {https://doi.org/10.48550/arXiv.2405.19304}}

@article{Richardson1968,
	_bib2doi_confirmed = {true},
	_bib2doi_finished = {true},
	_bib2doi_selected = {dblp:/rec/journals/jsyml/Richardson68.bib},
	author = {Richardson, Daniel},
	bibsource = {dblp computer science bibliography, https://dblp.org},
	biburl = {https://dblp.org/rec/journals/jsyml/Richardson68.bib},
	date-added = {2026-01-07 16:24:22 +0100},
	date-modified = {2026-01-07 16:24:22 +0100},
	doi = {10.2307/2271358},
	journal = {Journal of Symbolic Logic},
	number = {4},
	pages = {514--520},
	timestamp = {Wed, 14 Nov 2018 00:00:00 +0100},
	title = {Some Undecidable Problems Involving Elementary Functions of a Real Variable},
	volume = {33},
	year = {1968},
	bdsk-url-1 = {https://doi.org/10.2307/2271358}}

@book{Hartman,
	_bib2doi_finished = {true},
	author = {Hartman, Philip},
	date-added = {2026-01-07 16:24:22 +0100},
	date-modified = {2026-01-07 16:24:22 +0100},
	doi = {10.1137/1.9780898719222},
	file = {LuNancy.bib},
	publisher = {John Wiley and Sons},
	title = {Ordinary Differential Equations},
	year = {1964},
	bdsk-url-1 = {https://doi.org/10.1137/1.9780898719222}}

@article{Mi70,
	_bib2doi_confirmed = {true},
	_bib2doi_finished = {true},
	_bib2doi_selected = {dblp:/rec/journals/jcss/Miller70.bib},
	anciennecledoublon = {Miller70},
	author = {Miller, Webb},
	bibsource = {dblp computer science bibliography, https://dblp.org},
	biburl = {https://dblp.org/rec/journals/jcss/Miller70.bib},
	date-added = {2026-01-07 16:24:22 +0100},
	date-modified = {2026-01-07 16:24:22 +0100},
	doi = {10.1016/S0022-0000(70)80043-5},
	journal = {J. Comput. System Sci.},
	number = {5},
	pages = {465--472},
	timestamp = {Tue, 16 Feb 2021 00:00:00 +0100},
	title = {Recursive function theory and numerical analysis},
	url = {https://doi.org/10.1016/S0022-0000(70)80043-5},
	volume = {4},
	year = {1970},
	bdsk-url-1 = {https://doi.org/10.1016/S0022-0000(70)80043-5}}

@article{Abe71,
	_bib2doi_finished = {true},
	author = {Aberth, Oliver},
	date-added = {2026-01-07 16:24:22 +0100},
	date-modified = {2026-01-07 16:24:22 +0100},
	doi = {10.2307/2038240},
	journal = {Proceedings of the American Mathematical Society},
	pages = {151--156},
	title = {The failure in computable analysis of a classical existence theorem for differential equations},
	volume = {30},
	year = {1971},
	bdsk-url-1 = {https://doi.org/10.2307/2038240}}

@article{Ko83,
	_bib2doi_confirmed = {true},
	_bib2doi_finished = {true},
	_bib2doi_selected = {dblp:/rec/journals/iandc/Ko83.bib},
	author = {Ko, Ker{-}I},
	bibdate = {2006-04-25},
	bibsource = {dblp computer science bibliography, https://dblp.org},
	biburl = {https://dblp.org/rec/journals/iandc/Ko83.bib},
	date-added = {2026-01-07 16:38:12 +0100},
	date-modified = {2026-01-07 16:38:12 +0100},
	doi = {10.1016/S0019-9958(83)80062-X},
	journal = {Information and Control},
	month = {jul/aug/sep},
	number = {1-3},
	pages = {157--194},
	timestamp = {Fri, 12 Feb 2021 00:00:00 +0100},
	title = {On the Computational Complexity of Ordinary Differential Equations},
	volume = {58},
	year = {1983},
	bdsk-url-1 = {https://doi.org/10.1016/S0019-9958(83)80062-X}}

@article{Pap94b,
	_bib2doi_confirmed = {true},
	_bib2doi_finished = {true},
	_bib2doi_selected = {dblp:/rec/journals/jcss/Papadimitriou94.bib},
	author = {Papadimitriou, Christos H.},
	bibsource = {dblp computer science bibliography, https://dblp.org},
	biburl = {https://dblp.org/rec/journals/jcss/Papadimitriou94.bib},
	date-added = {2026-01-07 16:38:12 +0100},
	date-modified = {2026-01-07 16:38:12 +0100},
	doi = {10.1016/S0022-0000(05)80063-7},
	file = {TheorieDesJeuxCommunications.bib},
	journal = {Journal of Computer and System Sciences},
	month = {jun},
	number = {3},
	pages = {498--532},
	preliminary = {FOCS::Papadimitriou1990},
	timestamp = {Tue, 16 Feb 2021 00:00:00 +0100},
	title = {On the Complexity of the Parity Argument and Other Inefficient Proofs of Existence},
	url = {https://doi.org/10.1016/S0022-0000(05)80063-7},
	volume = {48},
	year = {1994},
	bdsk-url-1 = {https://doi.org/10.1016/S0022-0000(05)80063-7}}

@inproceedings{Var82,
	_bib2doi_confirmed = {true},
	_bib2doi_finished = {true},
	_bib2doi_selected = {dblp:/rec/conf/stoc/Vardi82.bib},
	author = {Vardi, Moshe Y.},
	bibsource = {dblp computer science bibliography, https://dblp.org},
	biburl = {https://dblp.org/rec/conf/stoc/Vardi82.bib},
	booktitle = {Proceedings of the 14th Annual {ACM} Symposium on Theory of Computing, May 5-7, 1982, San Francisco, California, {USA}},
	date-added = {2026-01-07 16:38:12 +0100},
	date-modified = {2026-01-07 16:38:12 +0100},
	doi = {10.1145/800070.802186},
	editor = {Lewis, Harry R. and Simons, Barbara B. and Burkhard, Walter A. and Landweber, Lawrence H.},
	pages = {137--146},
	publisher = {{ACM}},
	timestamp = {Wed, 14 Nov 2018 00:00:00 +0100},
	title = {The Complexity of Relational Query Languages (Extended Abstract)},
	year = {1982},
	bdsk-url-1 = {https://doi.org/10.1145/800070.802186}}

@inproceedings{PV,
	_bib2doi_confirmed = {true},
	_bib2doi_finished = {true},
	_bib2doi_selected = {dblp:/rec/conf/hybrid/PuriV94.bib},
	acknowledgement = {ack-nhfb},
	author = {Puri, Anuj and Varaiya, Pravin},
	bdsk-url-1 = {https://doi.org/10.1007/3-540-60472-3\\\_18},
	bibdate = {Sat May 11 13:45:32 MDT 1996},
	bibsource = {dblp computer science bibliography, https://dblp.org},
	biburl = {https://dblp.org/rec/conf/hybrid/PuriV94.bib},
	booktitle = {Hybrid Systems II, Proceedings of the Third International Workshop on Hybrid Systems, Ithaca, NY, USA, October 1994},
	coden = {LNCSD9},
	date-added = {2026-01-07 16:41:46 +0100},
	date-modified = {2026-01-07 16:41:46 +0100},
	doi = {10.1007/3-540-60472-3\\\\_18},
	editor = {Antsaklis, Panos J. and Kohn, Wolf and Nerode, Anil and Sastry, Shankar},
	file = {HSII.bib},
	issn = {0302-9743},
	pages = {359--369},
	publisher = {Springer},
	series = {Lecture Notes in Computer Science},
	timestamp = {Sat, 20 May 2017 01:00:00 +0200},
	title = {Verification of Hybrid Systems Using Abstractions},
	url = {https://doi.org/10.1007/3-540-60472-3\\\\_18},
	volume = {999},
	year = {1994},
	bdsk-url-1 = {https://doi.org/10.1007/3-540-60472-3%5C%5C%5C%5C_18}}

@article{Ruo96,
	_bib2doi_confirmed = {true},
	_bib2doi_finished = {true},
	_bib2doi_selected = {dblp:/rec/journals/ijfcs/Ruohonen96.bib},
	author = {Ruohonen, Keijo},
	bibsource = {dblp computer science bibliography, https://dblp.org},
	biburl = {https://dblp.org/rec/journals/ijfcs/Ruohonen96.bib},
	date-added = {2026-01-07 16:41:46 +0100},
	date-modified = {2026-01-07 16:41:46 +0100},
	doi = {10.1142/S0129054196000129},
	journal = {International Journal of Foundations of Computer Science},
	number = {2},
	pages = {151--160},
	timestamp = {Sat, 27 May 2017 01:00:00 +0200},
	title = {An Effective Cauchy-Peano Existence Theorem for Unique Solutions},
	url = {https://doi.org/10.1142/S0129054196000129},
	volume = {7},
	year = {1996},
	bdsk-url-1 = {https://doi.org/10.1142/S0129054196000129}}

@inproceedings{dsg05,
	_bib2doi_confirmed = {true},
	_bib2doi_finished = {true},
	_bib2doi_selected = {dblp:/rec/conf/cie/GracaCB05.bib},
	author = {Gra{\c{c}}a, Daniel Silva and Campagnolo, Manuel Lameiras and Buescu, Jorge},
	bibsource = {dblp computer science bibliography, https://dblp.org},
	biburl = {https://dblp.org/rec/conf/cie/GracaCB05.bib},
	booktitle = {New Computational Paradigms, First Conference on Computability in Europe, CiE 2005, Amsterdam, The Netherlands, June 8-12, 2005, Proceedings},
	date-added = {2026-01-07 16:44:24 +0100},
	date-modified = {2026-01-07 16:44:24 +0100},
	doi = {10.1007/11494645\\\_21},
	editor = {Cooper, S. Barry and L{\"{o}}we, Benedikt and Torenvliet, Leen},
	file = {analog.bib},
	pages = {169--179},
	pdf = {http://wslc.math.ist.utl.pt/ftp/pub/GracaDS/05-GCB-stable.pdf},
	proj = {contcomp},
	ps = {http://wslc.math.ist.utl.pt/ftp/pub/GracaDS/05-GCB-stable.ps},
	publisher = {Springer-Verlag},
	series = {Lecture Notes in Computer Science},
	st = {a},
	timestamp = {Wed, 29 Jul 2020 01:00:00 +0200},
	title = {Robust Simulations of Turing Machines with Analytic Maps and Flows},
	url = {https://doi.org/10.1007/11494645\\\\_21},
	volume = {3526},
	year = {2005},
	bdsk-url-1 = {https://doi.org/10.1007/11494645%5C%5C%5C%5C_21},
	bdsk-url-2 = {https://doi.org/10.1007/11494645%5C%5C%5C_21}}

@article{dsg06a,
	_bib2doi_finished = {true},
	author = {Gra{\c c}a, Daniel S. and Zhong, N. and Buescu, J.},
	date-added = {2026-01-07 16:44:24 +0100},
	date-modified = {2026-01-07 16:44:24 +0100},
	doi = {10.1090/s0002-9947-09-04929-0},
	journal = {Transactions of the American Mathematical Society},
	number = {6},
	pages = {2913--2927},
	pdf = {http://wslc.math.ist.utl.pt/ftp/pub/GracaDS/06-GZB-noncompode.pdf},
	proj = {contcomp,t3},
	st = {a},
	title = {Computability, noncomputability and undecidability of maximal intervals of {IVP}s},
	volume = {361},
	year = {2009},
	bdsk-url-1 = {https://doi.org/10.1090/s0002-9947-09-04929-0}}

@book{krajicek1995bounded,
    author = "Kraj{\'{\i}}cek, Jan",
    _bib2doi_confirmed = "true",
    _bib2doi_finished = "true",
    _bib2doi_selected = "dblp:/rec/books/daglib/0082595.bib",
    bibsource = "dblp computer science bibliography, https://dblp.org",
    biburl = "https://dblp.org/rec/books/daglib/0082595.bib",
    date-added = "2026-01-07 16:48:16 +0100",
    date-modified = "2026-01-07 16:48:16 +0100",
    doi = "10.1017/cbo9780511529948",
    isbn = "978-0-521-45205-2",
    publisher = "Cambridge University Press",
    series = "Encyclopedia of mathematics and its applications",
    timestamp = "Fri, 19 Jul 2019 01:00:00 +0200",
    title = "Bounded arithmetic, propositional logic, and complexity theory",
    volume = "60",
    year = "1995",
    bdsk-url-1 = "https://doi.org/10.1017/cbo9780511529948"
}

@book{Ko91,
    author = "Ko, Ker{-}I",
    _bib2doi_confirmed = "true",
    _bib2doi_finished = "true",
    _bib2doi_selected = "dblp:/rec/books/daglib/0067010.bib",
    address = "Boston",
    bibsource = "dblp computer science bibliography, https://dblp.org",
    biburl = "https://dblp.org/rec/books/daglib/0067010.bib",
    date-added = "2026-01-07 16:48:16 +0100",
    date-modified = "2026-01-07 16:48:16 +0100",
    doi = "10.1007/978-1-4684-6802-1",
    file = "analog.bib",
    isbn = "978-0-8176-3586-2",
    publisher = {Birkh{\"{a}}user},
    series = "Progress in theoretical computer science",
    timestamp = "Wed, 17 Jul 2019 01:00:00 +0200",
    title = "Complexity theory of real functions",
    volume = "3",
    year = "1991",
    bdsk-url-1 = "https://doi.org/10.1007/978-1-4684-6802-1"
}

@book{Wei00,
    author = "Weihrauch, Klaus",
    _bib2doi_confirmed = "true",
    _bib2doi_finished = "true",
    _bib2doi_selected = "dblp:/rec/series/txtcs/Weihrauch00.bib",
    anciennecledoublon = "LivreComputableAnalysis",
    bibsource = "dblp computer science bibliography, https://dblp.org",
    biburl = "https://dblp.org/rec/series/txtcs/Weihrauch00.bib",
    date-added = "2026-01-07 16:48:16 +0100",
    date-modified = "2026-01-07 16:48:16 +0100",
    doi = "10.1007/978-3-642-56999-9",
    isbn = "978-3-540-66817-6",
    publisher = "Springer",
    series = "Texts in Theoretical Computer Science. An {EATCS} Series",
    timestamp = "Tue, 16 May 2017 01:00:00 +0200",
    title = "Computable Analysis - An Introduction",
    year = "2000",
    bdsk-url-1 = "https://doi.org/10.1007/978-3-642-56999-9"
}

@inproceedings{Daskalakis06,
    author = "Daskalakis, Constantinos and Papadimitriou, Christos H.",
    editor = "Feigenbaum, Joan and Chuang, John C.{-}I. and Pennock, David M.",
    _bib2doi_confirmed = "true",
    _bib2doi_finished = "true",
    _bib2doi_old_doi = "http://doi.acm.org/10.1145/1134707.1134718",
    _bib2doi_selected = "dblp:/rec/conf/sigecom/DaskalakisP06.bib",
    address = "New York, NY, USA",
    bibsource = "dblp computer science bibliography, https://dblp.org",
    biburl = "https://dblp.org/rec/conf/sigecom/DaskalakisP06.bib",
    booktitle = "Proceedings 7th {ACM} Conference on Electronic Commerce (EC-2006), Ann Arbor, Michigan, USA, June 11-15, 2006",
    date-added = "2026-01-07 16:48:16 +0100",
    date-modified = "2026-01-07 16:48:16 +0100",
    doi = "10.1145/1134707.1134718",
    isbn = "1-59593-236-4",
    location = "Ann Arbor, Michigan, USA",
    pages = "91--99",
    publisher = "{ACM}",
    timestamp = "Sun, 19 Jan 2025 00:00:00 +0100",
    title = "Computing pure nash equilibria in graphical games via markov random fields",
    url = "https://doi.org/10.1145/1134707.1134718",
    year = "2006",
    bdsk-url-1 = "http://doi.acm.org/10.1145/1134707.1134718"
}

@inproceedings{collins2008effectivesimpl,
    author = "Collins, Pieter and Gra{\c{c}}a, Daniel Silva",
    editor = "Brattka, Vasco and Dillhage, Ruth and Grubba, Tanja and Klutsch, Angela",
    _bib2doi_confirmed = "true",
    _bib2doi_finished = "true",
    _bib2doi_selected = "dblp:/rec/journals/entcs/CollinsG08.bib",
    bibsource = "dblp computer science bibliography, https://dblp.org",
    biburl = "https://dblp.org/rec/journals/entcs/CollinsG08.bib",
    booktitle = "Proceedings of the Fifth International Conference on Computability and Complexity in Analysis, {CCA} 2008, Hagen, Germany, August 21-24, 2008",
    date-added = "2026-01-07 16:47:18 +0100",
    date-modified = "2026-01-07 16:47:18 +0100",
    doi = "10.1016/j.entcs.2008.12.010",
    journal = "Electronic Notes in Theoretical Computer Science",
    pages = "103--114",
    publisher = "Elsevier",
    series = "Electronic Notes in Theoretical Computer Science",
    timestamp = "Fri, 17 Feb 2023 00:00:00 +0100",
    title = "Effective Computability of Solutions of Ordinary Differential Equations The Thousand Monkeys Approach",
    url = "https://doi.org/10.1016/j.entcs.2008.12.010",
    volume = "221",
    year = "2008",
    bdsk-url-1 = "https://doi.org/10.1016/j.entcs.2008.12.010"
}

@inproceedings{kawamura2009lipschitz,
    author = "Kawamura, Akitoshi",
    _bib2doi_confirmed = "true",
    _bib2doi_finished = "true",
    _bib2doi_old_doi = "10.1109/ccc.2009.34",
    _bib2doi_selected = "dblp:/rec/conf/coco/Kawamura09.bib",
    bibsource = "dblp computer science bibliography, https://dblp.org",
    biburl = "https://dblp.org/rec/conf/coco/Kawamura09.bib",
    booktitle = "Proceedings of the 24th Annual {IEEE} Conference on Computational Complexity, {CCC} 2009, Paris, France, 15-18 July 2009",
    date-added = "2026-01-07 16:47:18 +0100",
    date-modified = "2026-01-07 16:47:18 +0100",
    doi = "10.1109/CCC.2009.34",
    organization = "IEEE",
    pages = "149--160",
    publisher = "{IEEE} Computer Society",
    timestamp = "Fri, 24 Mar 2023 00:00:00 +0100",
    title = "Lipschitz Continuous Ordinary Differential Equations are Polynomial-Space Complete",
    url = "https://doi.org/10.1109/CCC.2009.34",
    year = "2009",
    bdsk-url-1 = "https://doi.org/10.1109/ccc.2009.34"
}

@article{Ruo97b,
    author = "Ruohonen, Keijo",
   _bib2doi_old_doi = "10.1002/(sici)1099-0526(199707/08)2:6<41::aid-cplx9>3.0.co;2-k",
    _bib2doi_selected = "dblp:/rec/journals/complexity/Ruohonen97.bib",
    bibsource = "dblp computer science bibliography, https://dblp.org",
    biburl = "https://dblp.org/rec/journals/complexity/Ruohonen97.bib",
    date-added = "2026-01-07 16:47:18 +0100",
    date-modified = "2026-01-07 16:47:18 +0100",
     journal = "Complexity",
    number = "6",
    pages = "41--53",
    timestamp = "Thu, 24 Sep 2020 01:00:00 +0200",
    title = "Decidability and complexity of event detection problems for ODEs",
    volume = "2",
    year = "1997",
}

@article{PoulyGraca16,
	_bib2doi_confirmed = {true},
	_bib2doi_finished = {true},
	_bib2doi_selected = {dblp:/rec/journals/tcs/PoulyG16.bib},
	author = {Pouly, Amaury and Gra{\c{c}}a, Daniel Silva},
	bibsource = {dblp computer science bibliography, https://dblp.org},
	biburl = {https://dblp.org/rec/journals/tcs/PoulyG16.bib},
	date-added = {2026-01-07 17:24:10 +0100},
	date-modified = {2026-01-07 17:24:10 +0100},
	doi = {10.1016/j.tcs.2016.02.002},
	journal = {Theoretical Computer Science},
	pages = {67--82},
	timestamp = {Wed, 17 Feb 2021 00:00:00 +0100},
	title = {Computational complexity of solving polynomial differential equations over unbounded domains},
	url = {https://doi.org/10.1016/j.tcs.2016.02.002},
	volume = {626},
	year = {2016},
	bdsk-url-1 = {https://doi.org/10.1016/j.tcs.2016.02.002}}

@article{kawamura2015computational,
	_bib2doi_confirmed = {true},
	_bib2doi_finished = {true},
	_bib2doi_selected = {dblp:/rec/journals/jc/KawamuraMR015.bib},
	author = {Kawamura, Akitoshi and M{\"{u}}ller, Norbert Th. and R{\"{o}}snick, Carsten and Ziegler, Martin},
	bibsource = {dblp computer science bibliography, https://dblp.org},
	biburl = {https://dblp.org/rec/journals/jc/KawamuraMR015.bib},
	date-added = {2026-01-07 17:24:10 +0100},
	date-modified = {2026-01-07 17:24:10 +0100},
	doi = {10.1016/j.jco.2015.05.001},
	journal = {Journal of Complexity},
	number = {5},
	pages = {689--714},
	publisher = {Elsevier},
	timestamp = {Thu, 20 Feb 2020 00:00:00 +0100},
	title = {Computational benefit of smoothness: Parameterized bit-complexity of numerical operators on analytic functions and Gevrey's hierarchy},
	url = {https://doi.org/10.1016/j.jco.2015.05.001},
	volume = {31},
	year = {2015},
	bdsk-url-1 = {https://doi.org/10.1016/j.jco.2015.05.001}}

@article{Simpson84Peano,
	_bib2doi_confirmed = {true},
	_bib2doi_finished = {true},
	_bib2doi_selected = {dblp:/rec/journals/jsyml/Simpson84.bib},
	author = {Simpson, Stephen G.},
	bibsource = {dblp computer science bibliography, https://dblp.org},
	biburl = {https://dblp.org/rec/journals/jsyml/Simpson84.bib},
	date-added = {2026-01-07 17:24:10 +0100},
	date-modified = {2026-01-07 17:24:10 +0100},
	doi = {10.2307/2274131},
	journal = {The Journal of Symbolic Logic},
	number = {3},
	pages = {783--802},
	publisher = {Cambridge University Press},
	timestamp = {Sun, 06 Oct 2024 01:00:00 +0200},
	title = {Which Set Existence Axioms are Needed to Prove the Cauchy/Peano Theorem for Ordinary Differential Equations?},
	url = {https://doi.org/10.2307/2274131},
	volume = {49},
	year = {1984},
	bdsk-url-1 = {https://doi.org/10.2307/2274131}}

@inproceedings{kawamura2018parameterized,
	_bib2doi_confirmed = {true},
	_bib2doi_finished = {true},
	_bib2doi_selected = {dblp:/rec/conf/wollic/KawamuraST18.bib},
	author = {Kawamura, Akitoshi and Steinberg, Florian and Thies, Holger},
	bibsource = {dblp computer science bibliography, https://dblp.org},
	biburl = {https://dblp.org/rec/conf/wollic/KawamuraST18.bib},
	booktitle = {Logic, Language, Information, and Computation - 25th International Workshop, WoLLIC 2018, Bogota, Colombia, July 24-27, 2018, Proceedings},
	date-added = {2026-01-07 17:24:10 +0100},
	date-modified = {2026-01-07 17:24:10 +0100},
	doi = {10.1007/978-3-662-57669-4\_13},
	editor = {Moss, Lawrence S. and de Queiroz, Ruy J. G. B. and Mart{\'{\i}}nez, Maricarmen},
	organization = {Springer},
	pages = {223--236},
	publisher = {Springer},
	series = {Lecture Notes in Computer Science},
	timestamp = {Mon, 02 Nov 2020 00:00:00 +0100},
	title = {Parameterized Complexity for Uniform Operators on Multidimensional Analytic Functions and {ODE} Solving},
	url = {https://doi.org/10.1007/978-3-662-57669-4\\_13},
	volume = {10944},
	year = {2018},
	bdsk-url-1 = {https://doi.org/10.1007/978-3-662-57669-4%5C%5C_13},
	bdsk-url-2 = {https://doi.org/10.1007/978-3-662-57669-4%5C_13}}

@article{kawamura2014computational,
	_bib2doi_confirmed = {true},
	_bib2doi_finished = {true},
	_bib2doi_old_doi = {10.2168/lmcs-10(1:6)2014},
	_bib2doi_selected = {dblp:/rec/journals/corr/KawamuraORZ13.bib},
	author = {Kawamura, Akitoshi and Ota, Hiroyuki and R{\"{o}}snick, Carsten and Ziegler, Martin},
	bibsource = {dblp computer science bibliography, https://dblp.org},
	biburl = {https://dblp.org/rec/journals/corr/KawamuraORZ13.bib},
	date-added = {2026-01-07 17:24:10 +0100},
	date-modified = {2026-01-07 17:24:10 +0100},
	doi = {10.2168/LMCS-10(1:6)2014},
	journal = {Log. Methods Comput. Sci.},
	number = {1},
	publisher = {Episciences. org},
	timestamp = {Thu, 25 Jun 2020 01:00:00 +0200},
	title = {Computational Complexity of Smooth Differential Equations},
	url = {https://doi.org/10.2168/LMCS-10(1:6)2014},
	volume = {10},
	year = {2014},
	bdsk-url-1 = {https://doi.org/10.2168/lmcs-10(1:6)2014}}

@book{Imm99,
	_bib2doi_confirmed = {true},
	_bib2doi_finished = {true},
	_bib2doi_selected = {dblp:/rec/books/daglib/0095988.bib},
	author = {Immerman, Neil},
	bibsource = {dblp computer science bibliography, https://dblp.org},
	biburl = {https://dblp.org/rec/books/daglib/0095988.bib},
	date-added = {2026-01-07 17:24:10 +0100},
	date-modified = {2026-01-07 17:24:10 +0100},
	doi = {10.1007/978-1-4612-0539-5},
	file = {ComplexiteImplicite.bib},
	isbn = {978-1-4612-6809-3},
	publisher = {Springer},
	series = {Graduate texts in computer science},
	timestamp = {Tue, 16 May 2017 01:00:00 +0200},
	title = {Descriptive complexity},
	url = {https://doi.org/10.1007/978-1-4612-0539-5},
	year = {1999},
	bdsk-url-1 = {https://doi.org/10.1007/978-1-4612-0539-5}}

@book{cook2010logical,
	_bib2doi_finished = {true},
	author = {Cook, Stephen and Nguyen, Phuong},
	date-added = {2026-01-07 17:24:10 +0100},
	date-modified = {2026-01-07 17:24:10 +0100},
	doi = {10.1017/cbo9780511676277},
	publisher = {Cambridge University Press Cambridge},
	title = {Logical foundations of proof complexity},
	volume = {11},
	year = {2010},
	bdsk-url-1 = {https://doi.org/10.1017/cbo9780511676277}}

@book{simpson2009subsystems,
	_bib2doi_finished = {true},
	author = {Simpson, Stephen George},
	date-added = {2026-01-07 17:24:10 +0100},
	date-modified = {2026-04-20 16:48:02 +0200},
	doi = {10.1017/cbo9780511581007},
	publisher = {Cambridge University Press},
	title = {Subsystems of second order arithmetic},
	volume = {1},
	year = {2009},
	bdsk-url-1 = {https://doi.org/10.1017/cbo9780511581007}}

@incollection{Cob65,
	_bib2doi_finished = {true},
	address = {Amsterdam},
	author = {Cobham, Alan},
	booktitle = {Proceedings of the International Conference on Logic, Methodology, and Philosophy of Science},
	date-added = {2024-04-24 21:08:05 +0200},
	date-modified = {2024-04-24 21:08:05 +0200},
	editor = {Bar-Hillel, Y.},
	file = {BSSfunc.bib},
	pages = {24--30},
	publisher = {North-Holland},
	title = {The intrinsic computational difficulty of functions},
	year = {1962}}

@book{CL72,
	_bib2doi_finished = {true},
	author = {Coddington, Earl E. and Levinson, Norman},
	date-added = {2024-04-24 21:08:05 +0200},
	date-modified = {2024-04-24 21:08:05 +0200},
	file = {LuEns.bib},
	month = {apr},
	publisher = {McGraw-Hill},
	title = {Theory of Ordinary Differentiel Equations},
	year = {1972}}

@book{bostan2017algorithmes,
	_bib2doi_finished = {true},
	author = {Bostan, Alin and Chyzak, Fr{\'e}d{\'e}ric and Giusti, Marc and Lebreton, Romain and Lecerf, Gr{\'e}goire and Salvy, Bruno and Schost, {\'E}ric},
	date-added = {2026-01-08 17:40:14 +0100},
	date-modified = {2026-01-08 17:40:14 +0100},
	publisher = {published by the Authors},
	title = {Algorithmes efficaces en calcul formel},
	year = {2017}}

@phdthesis{TheseRiccardo,
	_bib2doi_finished = {true},
	author = {Gozzi, Riccardo},
	date-added = {2026-01-08 17:40:14 +0100},
	date-modified = {2026-01-08 17:40:14 +0100},
	school = {Instituto Superior T{\'e}cnico, Lisbon, Portugal and University of Algarve, Faro, Portugal},
	title = {Analog Characterization of Complexity Classes},
	year = {2022}}

@phdthesis{TheseAmaury,
	_bib2doi_finished = {true},
	author = {Pouly, Amaury},
	date-added = {2026-01-08 17:40:14 +0100},
	date-modified = {2026-01-08 17:40:14 +0100},
	month = {Defended on July 6, 2015.},
	note = {https://pastel.archives-ouvertes.fr/tel-01223284, Prix de Th{\`e}se de l'Ecole Polyechnique 2016, Ackermann Award 2017},
	school = {Ecole Polytechnique and Unidersidade Do Algarve},
	title = {Continuous models of computation: from computability to complexity},
	year = {2015}}

@phdthesis{TheseManon,
	_bib2doi_finished = {true},
	author = {Blanc, Manon},
	date-added = {2025-06-06 10:11:26 +0200},
	date-modified = {2025-11-27 22:26:54 +0100},
	note = {Phd Young Talent Award L'Or{\'e}al-Unesco 2025, and Prix STIC Doctorants du plateau de Saclay 2025},
	school = {Ecole Polytechnique},
	title = {Discrete-Time and Continuous-Time Systems over the Reals: Relating Complexity with Robustness, Length and Precision},
	year = {2025}}

@incollection{GracaZhongHandbook,
	_bib2doi_finished = {true},
	author = {Gra{\c c}a, Daniel S. and Zhong, Ning},
	title = {Computability of Differential Equations},
	date-added = {2024-07-16 11:08:40 +0200},
	date-modified = {2024-07-16 11:08:40 +0200},
	editor = {Brattka, Vasco and Hertling, Peter},
	journal = {To appear.},
	publisher = {Springer.},
	booktitle = {Handbook of Computability and Complexity in Analysis},
	topics = {Analog,Dynamical,GPAC},
	year = {2018}}

@article{brent1978fast,
	author = {Brent, Richard P and Kung, Hsiang T},
	journal = {Journal of the ACM (JACM)},
	number = {4},
	pages = {581--595},
	publisher = {ACM New York, NY, USA},
	title = {Fast algorithms for manipulating formal power series},
	volume = {25},
	year = {1978}}

@book{Soare2016,
	address = {Berlin, Heidelberg},
	author = {Soare, Robert I.},
	doi = {10.1007/978-3-642-31933-4},
	isbn = {978-3-642-31932-7},
	publisher = {Springer},
	series = {Theory and Applications of Computability},
	title = {Turing Computability: Theory and Applications},
	year = {2016}}

@incollection{brattka2021weihrauch,
	author = {Brattka, Vasco and Gherardi, Guido and Pauly, Arno},
	booktitle = {Handbook of computability and complexity in analysis},
	pages = {367--417},
	publisher = {Springer},
	title = {Weihrauch complexity in computable analysis},
	year = {2021}}

@inproceedings{Antonelli0K24,
	author = {Melissa Antonelli and Arnaud Durand and Juha Kontinen},
	bibsource = {dblp computer science bibliography, https://dblp.org},
	booktitle = {49th International Symposium on Mathematical Foundations of Computer Science, {MFCS} 2024, Bratislava, Slovakia, August 26-30, 2024},
	doi = {10.4230/LIPICS.MFCS.2024.10},
	editor = {Rastislav Kr{\'{a}}lovic and Anton{\'{\i}}n Kucera},
	pages = {10:1--10:18},
	publisher = {Schloss Dagstuhl - Leibniz-Zentrum f{\"{u}}r Informatik},
	series = {LIPIcs},
	title = {A New Characterization of {FAC$^0$} via Discrete Ordinary Differential Equations},
	url = {https://doi.org/10.4230/LIPIcs.MFCS.2024.10},
	volume = {306},
	year = {2024}}

@inproceedings{Antonelli0K25,
	author = {Melissa Antonelli and Arnaud Durand and Juha Kontinen},
	bibsource = {dblp computer science bibliography, https://dblp.org},
	booktitle = {50th International Symposium on Mathematical Foundations of Computer Science, {MFCS} 2025, Warsaw, Poland, August 25-29, 2025},
	doi = {10.4230/LIPICS.MFCS.2025.10},
	editor = {Pawel Gawrychowski and Filip Mazowiecki and Michal Skrzypczak},
	pages = {10:1--10:18},
	publisher = {Schloss Dagstuhl - Leibniz-Zentrum f{\"{u}}r Informatik},
	series = {LIPIcs},
	title = {Characterizing Small Circuit Classes from {FAC$^0$} to {FAC$^1$} via Discrete Ordinary Differential Equations},
	url = {https://doi.org/10.4230/LIPIcs.MFCS.2025.10},
	volume = {345},
	year = {2025}}

@article{AntonelliDK26,
	author = {Melissa Antonelli and Arnaud Durand and Juha Kontinen},
	bibsource = {dblp computer science bibliography, https://dblp.org},
	doi = {10.1016/J.TCS.2025.115655},
	journal = {Theor. Comput. Sci.},
	pages = {115655},
	title = {Towards new characterizations of small circuit classes via discrete ordinary differential equations},
	url = {https://doi.org/10.1016/j.tcs.2025.115655},
	volume = {1062},
	year = {2026}}

@article{jockusch1972degrees,
	author = {Jockusch, Carl and Soare, Robert},
	journal = {Pacific Journal of Mathematics},
	number = {3},
	pages = {605--616},
	publisher = {Mathematical Sciences Publishers},
	title = {Degrees of members of $\Pi$10 classes},
	volume = {40},
	year = {1972}}

@article{Liu12,
	author = {Jiayi Liu},
	bibsource = {dblp computer science bibliography, https://dblp.org},
	doi = {10.2178/JSL/1333566640},
	journal = {J. Symb. Log.},
	number = {2},
	pages = {609--620},
	title = {{$RT^2_2$ does not imply $WKL_0$}},
	url = {https://doi.org/10.2178/jsl/1333566640},
	volume = {77},
	year = {2012}}

@article{SeetapunS95,
	author = {David Seetapun and Theodore A. Slaman},
	bibsource = {dblp computer science bibliography, https://dblp.org},
	doi = {10.1305/NDJFL/1040136917},
	journal = {Notre Dame J. Formal Log.},
	number = {4},
	pages = {570--582},
	title = {On the Strength of Ramsey's Theorem},
	url = {https://doi.org/10.1305/ndjfl/1040136917},
	volume = {36},
	year = {1995}}

@article{friedman1976systems,
	author = {Friedman, Harvey M},
	journal = {J. Symb. Logic},
	pages = {557--559},
	title = {Systems on second order arithmetic with restricted induction i, ii},
	volume = {41},
	year = {1976}}

@inproceedings{friedman1975some,
	author = {Friedman, Harvey M},
	booktitle = {Proc. of ICM},
	organization = {Canadian Math. Congress},
	pages = {235--242},
	title = {Some systems of second order arithmetic and their use},
	volume = {1},
	year = {1975}}

@article{ReverseMathOfComplexityLowerBournds24,
	author = {Lijie Chen and Jiatu Li and Igor Carboni Oliveira},
	bibsource = {dblp computer science bibliography, https://dblp.org},
	eprint = {TR24-060},
	eprinttype = {ECCC},
	journal = {Electron. Colloquium Comput. Complex.},
	title = {Reverse Mathematics of Complexity Lower Bounds},
	url = {https://eccc.weizmann.ac.il/report/2024/060},
	volume = {{TR24-060}},
	year = {2024}}

@article{PR79,
	_bib2doi_finished = {true},
	author = {Pour-El, M. B. and Richards, J. I.},
	date-added = {2026-01-08 23:10:07 +0100},
	date-modified = {2026-01-08 23:10:07 +0100},
	doi = {10.1016/0003-4843(79)90021-4},
	journal = {Ann. Math. Logic},
	pages = {61--90},
	title = {A computable ordinary differential equation which possesses no computable solution},
	volume = {17},
	year = {1979},
	bdsk-url-1 = {https://doi.org/10.1016/0003-4843(79)90021-4}}

@article{collins2009effective,
	_bib2doi_confirmed = {true},
	_bib2doi_finished = {true},
	_bib2doi_selected = {dblp:/rec/journals/jucs/CollinsG09.bib},
	author = {Pieter Collins and Daniel Silva Gra{\c{c}}a},
	bibsource = {dblp computer science bibliography, https://dblp.org},
	biburl = {https://dblp.org/rec/journals/jucs/CollinsG09.bib},
	date-added = {2024-05-02 17:35:31 +0200},
	date-modified = {2024-05-02 17:35:31 +0200},
	doi = {10.3217/jucs-015-06-1162},
	journal = {Journal of Universal Computer Science},
	number = {6},
	pages = {1162--1185},
	publisher = {www. jucs. org},
	timestamp = {Thu, 07 Sep 2023 01:00:00 +0200},
	title = {Effective Computability of Solutions of Differential Inclusions The Ten Thousand Monkeys Approach},
	url = {https://doi.org/10.3217/jucs-015-06-1162},
	volume = {15},
	year = {2009},
	bdsk-url-1 = {https://doi.org/10.3217/jucs-015-06-1162}}

@book{PORI17,
	_bib2doi_finished = {true},
	author = {Pour-El, Marian B and Richards, J Ian},
	date-added = {2024-05-02 17:27:36 +0200},
	date-modified = {2024-05-02 17:27:36 +0200},
	doi = {10.1017/9781316717325},
	publisher = {Cambridge University Press},
	title = {Computability in analysis and physics},
	volume = {1},
	year = {2017},
	bdsk-url-1 = {https://doi.org/10.1017/9781316717325}}

@incollection{brattka2008tutorial,
	_bib2doi_finished = {true},
	author = {Brattka, Vasco and Hertling, Peter and Weihrauch, Klaus},
	booktitle = {New computational paradigms},
	date-added = {2024-05-01 21:37:32 +0200},
	date-modified = {2024-05-01 21:37:32 +0200},
	doi = {10.1007/978-0-387-68546-5_18},
	pages = {425--491},
	publisher = {Springer},
	title = {A tutorial on computable analysis},
	year = {2008},
	bdsk-url-1 = {https://doi.org/10.1007/978-0-387-68546-5_18}}

\end{document}